\documentclass[11pt,leqno]{article}

\usepackage{hyperref}
\usepackage{amsmath,amsfonts,amsthm,latexsym}
\usepackage{tikz,graphicx,subfigure}
\usetikzlibrary{arrows,decorations.markings}
\usetikzlibrary{arrows}
\usepgflibrary{decorations.markings}
\usepackage{overpic}
\usepackage{psfrag}
\usepackage{comment}
\newtheorem{theorem}{Theorem}[section]
\newtheorem{lemma}[theorem]{Lemma}

\newtheorem{proposition}[theorem]{Proposition}
\newtheorem{Remark}[theorem]{Remark}

\newtheorem{Example}[theorem]{Example}
\newtheorem{Assumption}[theorem]{Assumption}
\newtheorem{definition}[theorem]{Definition}
\newenvironment{remark}{\begin{Remark}\rm}{\end{Remark}}

\numberwithin{equation}{section}
\DeclareMathOperator{\Tr}{Tr}

\newcommand{\eps}{\varepsilon}
\newcommand{\R}{\mathbb{R}}
\newcommand{\C}{\mathbb{C}}
\newcommand{\var}{\mathop{\mathrm{Var}}}
\renewcommand{\Im}{\mathop{\mathrm{Im}}}
\newcommand{\EE}{\mathbb{E}}
\newcommand{\HH}{{\mathbb{H}_+}}
\newcommand{\N}{\mathbb{N}}
\newcommand{\Z}{\mathbb{Z}}

\newcommand{\OO}{\mathcal O}
\renewcommand{\Re}{\mathop \mathrm{Re}}
\renewcommand{\Im}{\mathop \mathrm{Im}}

\title{The Gaussian free field in an interlacing particle system with two jump rates}
\author{ Maurice Duits\footnote{Department of Mathematics, California Institute of Technology, 1200 E. California Blvd, Pasadena CA 91125. E-mail: mduits@caltech.edu }}
\date{}

\begin{document}

\maketitle
\begin{abstract}
We study the fluctuations of a random surface in a stochastic growth model on a system of interlacing
particles placed on a two dimensional  lattice.  There are two different types of particles, one with a low jump rate and the other with a  high jump rate. In the large time  limit, the random surface has a deterministic shape.   Due to the different jump rates, the limit shape and the domain on which it is defined are not smooth. The main result  is that the fluctuations of the random surface  are governed by the Gaussian free field. 
\end{abstract}


\section{Introduction}

The Gaussian free field is commonly assumed to be a universal field describing the fluctuations of random surfaces appearing in a wide class of models in statistical physics. However, rigorous proofs are only known for some particular integrable models. For example, interesting progress has been made on  dimer models on bipartite planar graphs (see \cite{KenyonLecture} for a survey and a list of references). 

In the present paper we  study  a random surface appearing in a  stochastic growth model  on a system of interlacing particles and show that its fluctuations are governed by the Gaussian free field. This work is partially inspired by  \cite{BoFe}, where the authors introduced a general model in 2+1 dimensions that connects the random surfaces of the type that occur in the dimer models, with the random growth of particle systems on a one dimensional lattice (e.g.  exclusion processes).   Here we will discuss   a  particular specialization of that model.

In the initial configuration the particles are placed on the grid $\Z\times \N$ and are densely packed in a triangular region as shown in the left picture in  Figure~\ref{fig:init}.  Each particle has an exponential clock and when the clock rings, the particle attempts to jump to the right by $1$. To ensure that the interlacing of the particles on subsequent horizontal levels is preserved, each attempted jump is subject to certain rules (specified later on, but in short: each particle is blocked by the particles below, but pushes the particles above). The central feature  of our model is the fact that we have two different types of particles: slow and fast. More precisely, we draw a horizontal separating  line and  equip the particles above that line with a higher jump rate than the particles below. As a consequence the particles above the separating line tend to move faster. See \cite{BoFe} for a thorough analysis for the situation in which all particles have the same jump rate.
\begin{figure}[t]\label{fig:init}
	\centering{
	\begin{tikzpicture}
		\node (0,0){
	\includegraphics[scale=0.35]{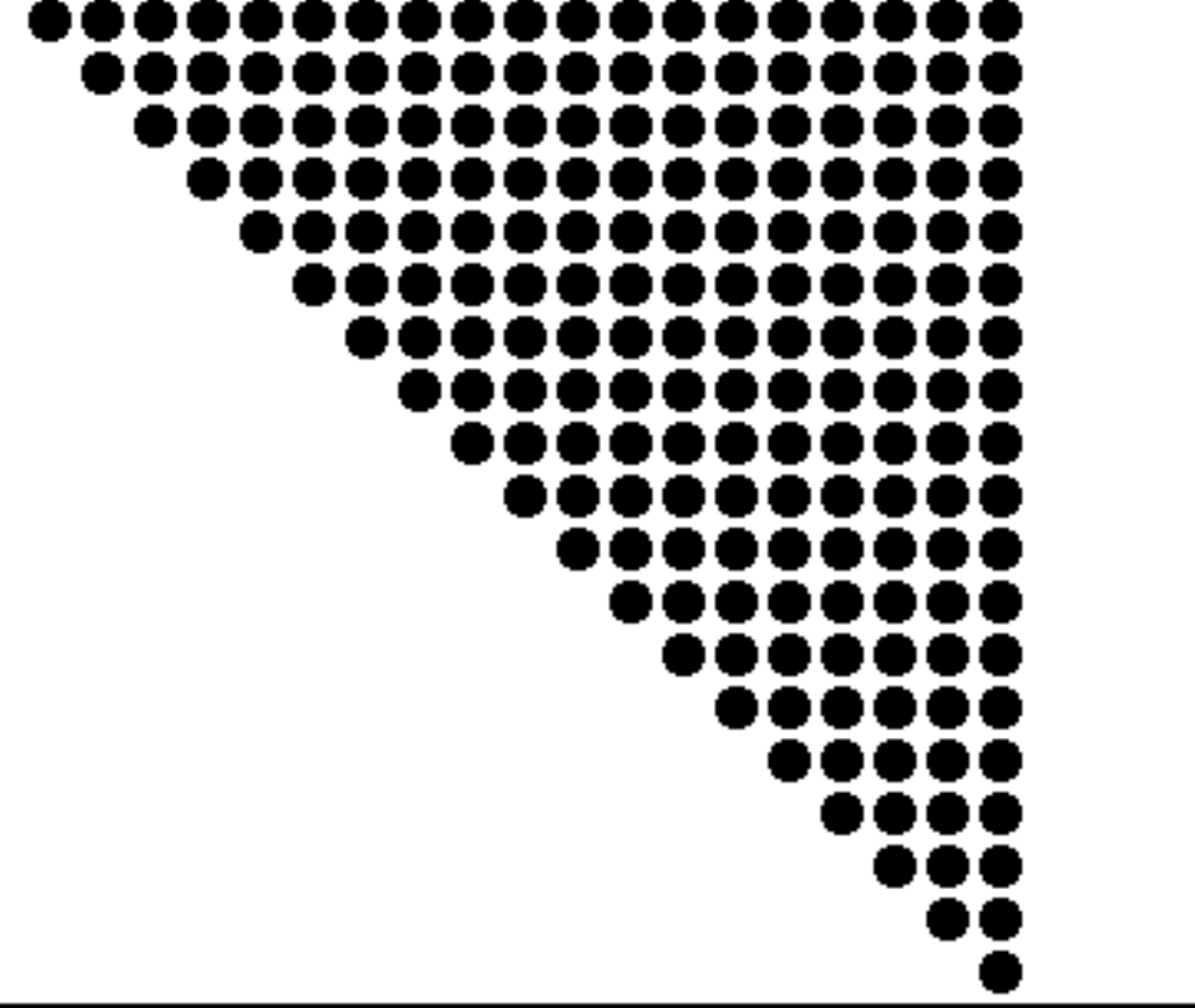}};
	\draw[-,dashed] (-2.3,0.31) --(2.3,0.31);
	\end{tikzpicture}\hspace{1cm}
	\begin{tikzpicture}
	\node (0,0) {\includegraphics[scale=0.35]{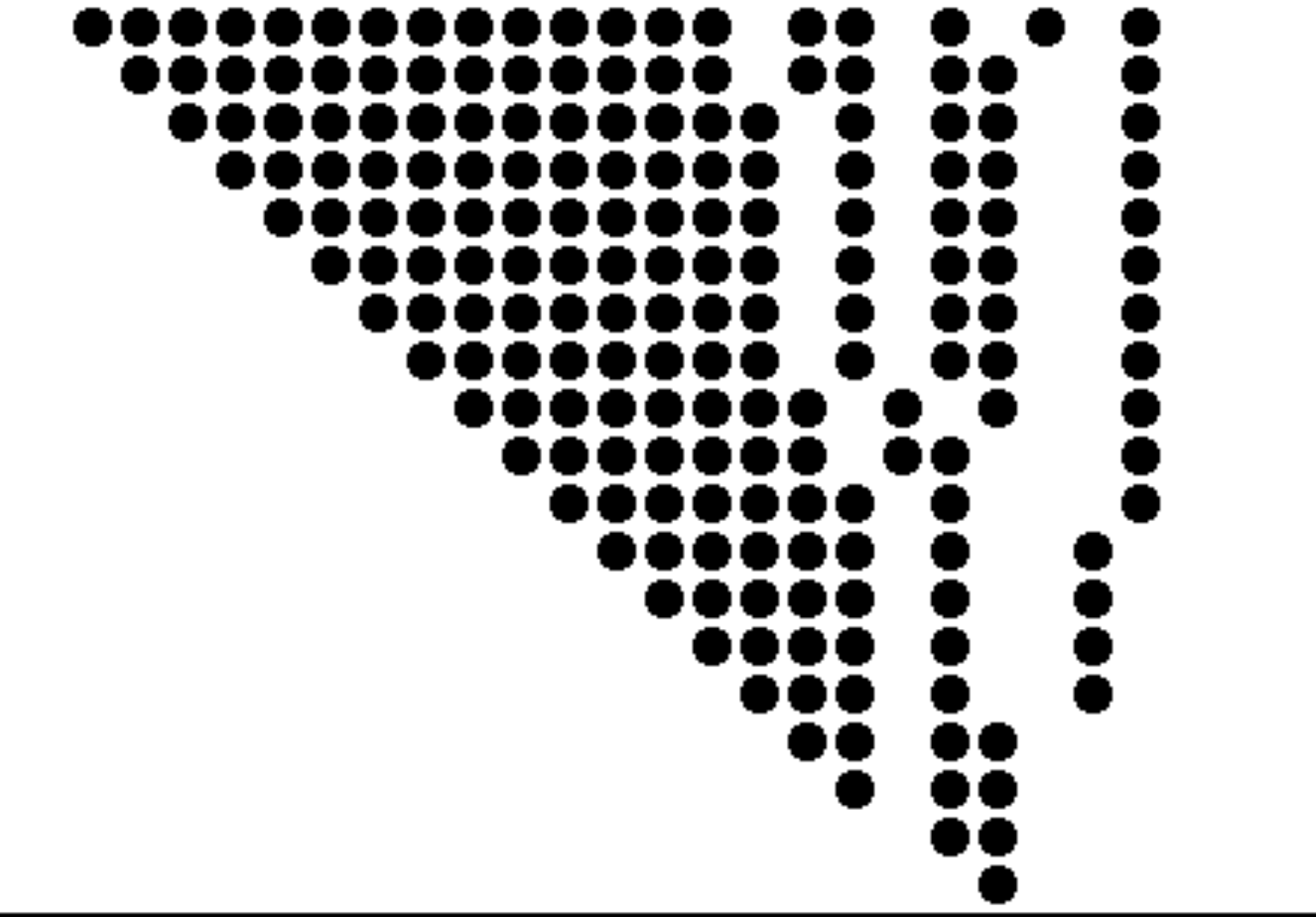}};
	\draw[-,dashed] (-2.8,0.31) --(2.8,0.31);
	\end{tikzpicture}}
	\caption{The left figure shows the initial configuration of the particles. The right figure shows a possible configuration after a short time. The dashed line separates the particles with a higher jump rate (above) from the particles with a lower jump rate (below). }
\end{figure}

 For large time, the particles will be distributed on a large domain that  is shown  in Figure~\ref{fig:domains}. See also Figures~\ref{fig:confbefore} and \ref{fig:confafter} for sample configurations. There exists a critical time, after which the limiting domain develops a cusp. This means that after the critical time, there is a group of fast particles that is drifting away from the slower particles below the separating line, while other fast particles are held back by the slow particles below due to the jump rules. The latter situation is illustrated in Figure~\ref{fig:confafter} and the right picture in Figure~\ref{fig:domains}. 
 
 Our main results on the long time behavior are expressed in terms of a height function that integrates the particle configuration. The graph of the height function defines the random surface that has our interest.   We show that this height function has a deterministic shape in the large time limit.  The difference in speed for the two types of particles induces a   jump discontinuity in the normal to the limit shape at the line separating the slow and fast particles.  

In order to describe the limit shape, we introduce a bijection that  maps the limiting domain to the upper half of the complex plane.  This map constitutes a natural complex structure on the system. Away from the separating line it satisfies the complex Burgers equation (for more details on the connection between this equation and limit shapes see \cite{KO}). Due to the different jump rates, the bijection in our case is only homeomorphism  but not a diffeomorphism (in contrast to the situations in for example \cite{BoFe,Kenyon}).  The main result of the paper is that the fluctuation of the pushforward of the random surface under this map are  governed by the Gaussian free field on the upper half plane.

The limit shape is subject to a   PDE. Away from the separating line, this PDE  brings our model \emph{locally}  in the 2D anisotropic  KPZ class \cite{BS} (see also \cite{BoFe} and the references therein).  This constitutes an important class of  stochastic growth models. In \cite{wolf}, the author used non-rigorous arguments to show that the fluctuations in these models are similar to the fluctuations of random surfaces in the Edwards-Wilkinson class, which predicts the logarithmic behavior of the height fluctuations.  However, it does not predict the full result on the Gaussian Free field, especially in relation with the complex structure. The only model of this type (that the author is aware of) for which the Gaussian free field in the fluctuations is rigorously proved,  is the one jump rate situation analyzed in \cite{BoFe}. In our situation, the PDE  has a jump discontinuity on the separating line. It is not a priori clear if and how this effects the fluctuations. For example,  it is not obvious at first (and perhaps slightly surprising), that the correlations of  the  fluctuations of the height function at a point in the slower part and at a point in the faster part that starts drifting away, are essentially of the same type as for  points that are both in the slow part. The essence of our main result is that in all cases these  correlations are given  by the pullback of the Green's function for the Laplace operator on the upper half plane with Dirichlet boundary conditions by the complex structure on the system.
   
 Finally, we remark that in for example \cite{BoFe,KenyonD,Kenyon,KenyonLecture},  the presence of the Gaussian free field is established by computing the limiting behavior for the moments of  the height fluctuations at multiple points.  In this paper, we follow an alternative approach that exploits the determinantal structure of the process.
  \paragraph{Acknowledgements} I thank  Alexei Borodin for drawing my attention to the subject of this paper   and for   many fruitful discussions.

\section{Statement of results}
 In this section we will state our main results. The proofs can be found in Sections \ref{sec:newco}-\ref{sec:asymptoticsR}.

\begin{figure}[t]
\centering{
\includegraphics[scale=0.3]{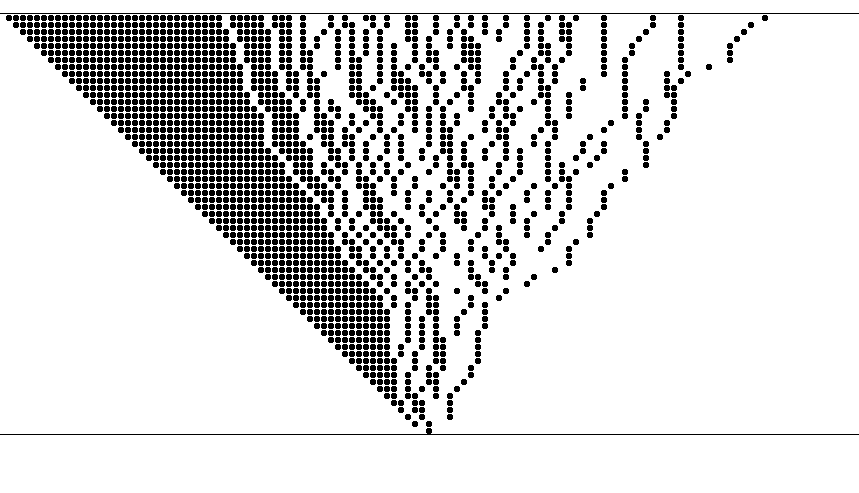}
}
\caption{A random particle configuration before the critical time. The difference in jump rates is exaggerated for illustration purposes.}
\label{fig:confbefore}
\end{figure}

\subsection{The model}

Let us start with describing  the model.  We consider dynamics on a system of interlacing particles placed on the grid $\Z\times \N$. At each point in time, there  are $m$ particles on the horizontal section $(\cdot,m)$ for $m\in \N$. We denote  the horizontal coordinate of the $k$-th particle (counting from left to right) on the $m$-th level  by $x_k^m$ for $k=1,\ldots,m$. In the initial configuration at time $t=0$ the positions are
\begin{align}
	x_k^m(0) =k-m-1, \quad \text{ for } k=1,\ldots m, 
\end{align}
as shown in the left picture of Figure~\ref{fig:init}. Each particle has an exponential clock and when its clock rings, it attempts to jump to the right by one. However, we enforce the following  interlacing condition to hold at all times
\begin{align}\label{eq:interl}
	x_{k}^{m+1} < x_{k}^m \leq x_{k+1}^{m+1}.
\end{align}  
To ensure that this interlacing condition holds,  we impose the following rules. If the exponential clock of the particles positioned at $x_k^m$ rings then
\begin{enumerate}
	\item it remains put if $x_k^{m-1}=x_k^m+1$. 
	\item it jumps to the right by one and so do all particles with horizontal coordinate $x_{k+l}^{m+l}=x_k^m$ for $l=1,2,\ldots$.
\end{enumerate} 
Hence a particle is blocked by particles that are below, but it pushes particles that lie above. 

It remains to set the rate for the exponential clocks. The first $m_0$ horizontal levels have an exponential clock with rate $1$, whereas the particles at the higher levels have rate $2$. Hence the particles at $x_k^m$ with $m\geq m_0$   will attempt to move faster in time. 

We are interested in the long time behavior of the system.

\begin{remark}
The above model can also be established as a tiling of the half plane by lozenges \cite{BoFe}. \end{remark}

 \subsection{Long time behavior}
 
 We proceed with an informal discussion on the long time behavior. Precise statements are formulated in the next paragraph.
 
 In Figures~\ref{fig:confbefore} and \ref{fig:confafter} we give two particular sample configurations. As these figures suggest, for large time  the particles fill a  domain  as given  in Figure~\ref{fig:domains}. The slanted dashed line in Figure~\ref{fig:domains} represents the line $x+m=0$. Note that  at all times, the particles are positioned at the right of that line.  The horizontal dashed line is the line $m=m_0$ that separates the slow and fast particles.   As a consequence  to the different speeds of the particles, there is a transition in the shape of the domain and the density as one crosses the line $m=m_0$. The particles above that line tend to move faster than the particles below. Due to the interlacing condition, the particles more to the left in the upper part of the system are held back  by the particles in the lower part. The particles to the right  move independently  from the lower part. There is a critical  time, after which  particles  more to the right in  the upper part start forming a  group that  starts drifting away and  a cusp appears in the domain. This is illustrated in the  domain at the right in Figure~\ref{fig:domains}. 
 
 \begin{figure}[t]
\centering{
\includegraphics[scale=0.3]{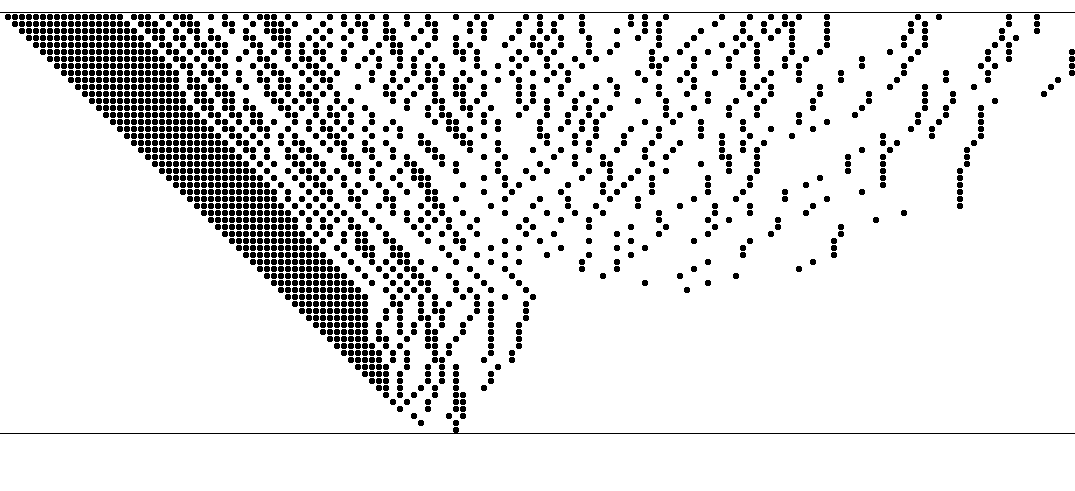}}
\caption{A random particle configuration after the critical time. The difference in jump rates is exaggerated for illustration purposes.}
\label{fig:confafter}
\end{figure}

Although the  focus to this paper is on the global scale and its fluctuations, it is also possible  retrieve interesting  universality    classes at the local scale. As the proofs are standard and the results are not relevant to this paper, we content ourselves with a brief discussion. For points in the bulk, the local correlation are described by extensions of the discrete sine kernel that fall into the class described in \cite{B2}. Near the edge, we obtain extensions of the Airy process (see \cite{PS} and \cite{Fer} for a review). The local correlations near the cusp are determined by the Pearcey process  \cite{ABK,BH,BH1,OR,TW}. At the points where the domain touches the line $x+m=0$, the horizontal axis and the line $m=m_0$ (after the critical time), the local process is governed by the GUE minor process \cite{BoDu,JN,OR0}.

\subsection{The limit shape}

To each particle configuration we assign  a  surface, which is the graph of the height function given in the following definition. \begin{definition} We define the height function $h:\Z\times \N \to \N_0$ by
\begin{align}\label{eq:defh}
	h(x,m)=\# \{x_k^m  \mid x_k^m\geq x\}.
\end{align}
In other words,  $h(x,m)$ counts  the number of particles that are at the right of $(x,m)$ (including the possible particle at $(x,m)$).
\end{definition}

For each configuration, the height function defines a stepped surface. Our first result is an explicit description of the limit shape of the mean  $\EE h$.   To this end, we need the following definitions.  

\begin{definition}
Denote the upper half plane by $\HH= \{z\in \C  \mid \Im z>0\}.$
Define the function $F:\HH \to \C$ by
\begin{multline}\label{eq:defF}
	F(z \mid \xi,\mu)=\\\left\{
	\begin{array}{lc}
		\tau z +\mu \log (z-1)-(\xi+\mu) \log z, &  \mu\leq \mu_0,\\
		\tau z+\mu_0\log(z-1)+(\mu-\mu_0)\log(z-2) -(\xi+\mu) \log z, & \mu>\mu_0,
	\end{array}
	\right.
\end{multline} 
for $z\in \HH$. To every $(\xi,\mu)$ there exists at most one $\Omega\in \HH$ such that $$F'(\Omega\mid  \mu, \xi)=0.$$ If it exists we denote it by $\Omega(\xi,\mu)$ and we define 
\begin{align}
	\mathcal D= \{(\xi,\mu)   \mid \Omega(\xi,\mu) \in \HH \text{ exists} \}.
\end{align}  
\end{definition}
The fact that to every $(\xi,\mu)$ there exists at most one $\Omega\in \HH$ such that $F'(\Omega \mid \mu, \xi)=0$, follows immediately after observing that $F'(\Omega \mid \xi,\mu)=0$ is equivalent to a quadratic (in case $\mu\leq \mu_0$) or cubic (in case $\mu>\mu_0$) equation in $\Omega$ with real coefficients. Moreover, we have the following.
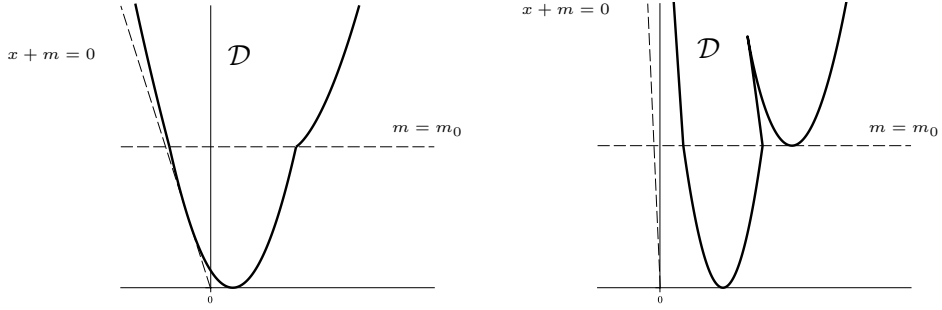
\begin{figure}[t]	\centering{
	\begin{tikzpicture} 
	\node (0,0) {\includegraphics[scale=0.22]{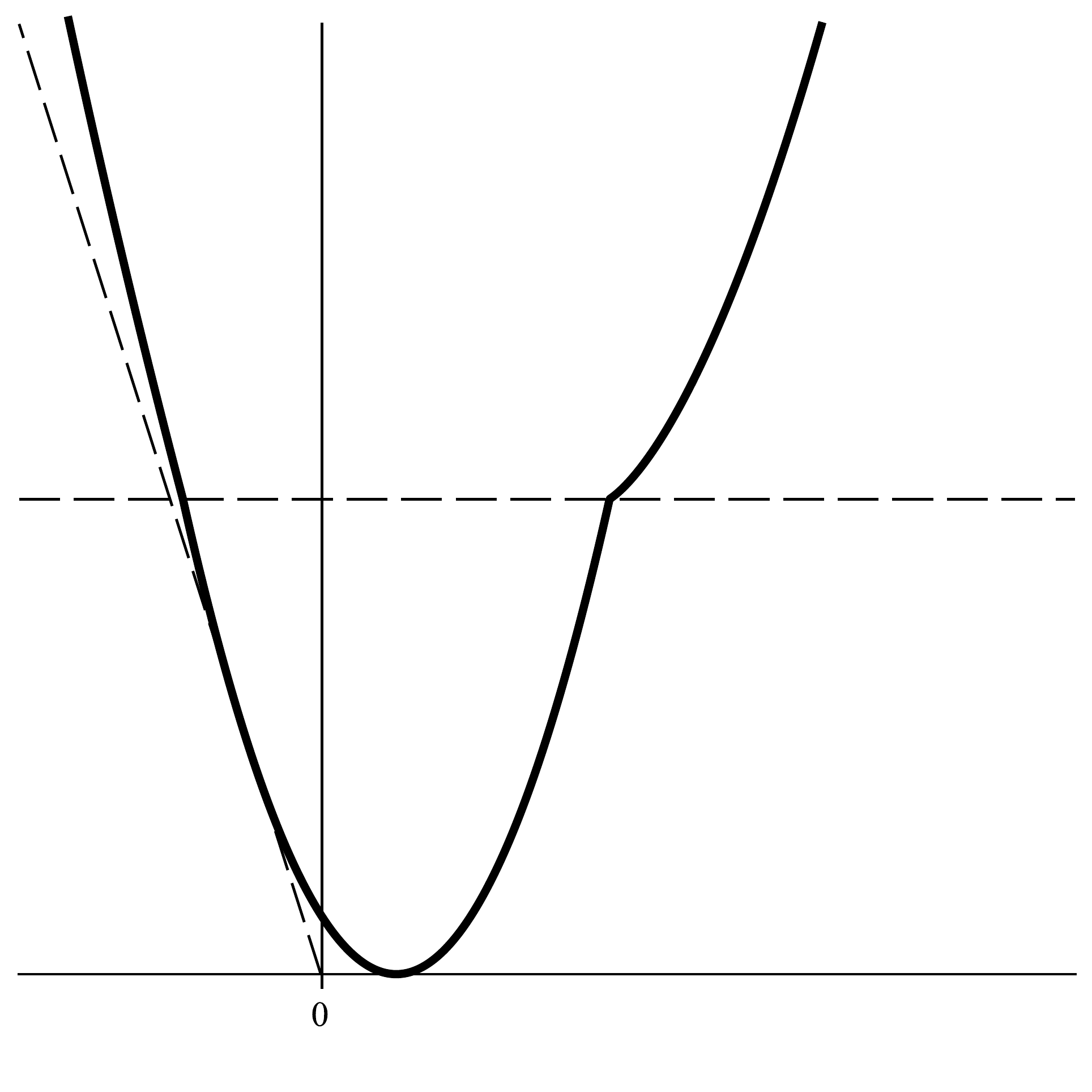}};
	\draw (-3,1.4) node {\tiny{$x+m=0$}};
	\draw (2,0.4) node {\tiny{$m=m_0$}};
		\draw (-0.5,1.4) node {$\mathcal D$};
	\end{tikzpicture}\hspace{0.4cm}
	\begin{tikzpicture} 
	\node (0,0) {\includegraphics[grid,scale=0.22]{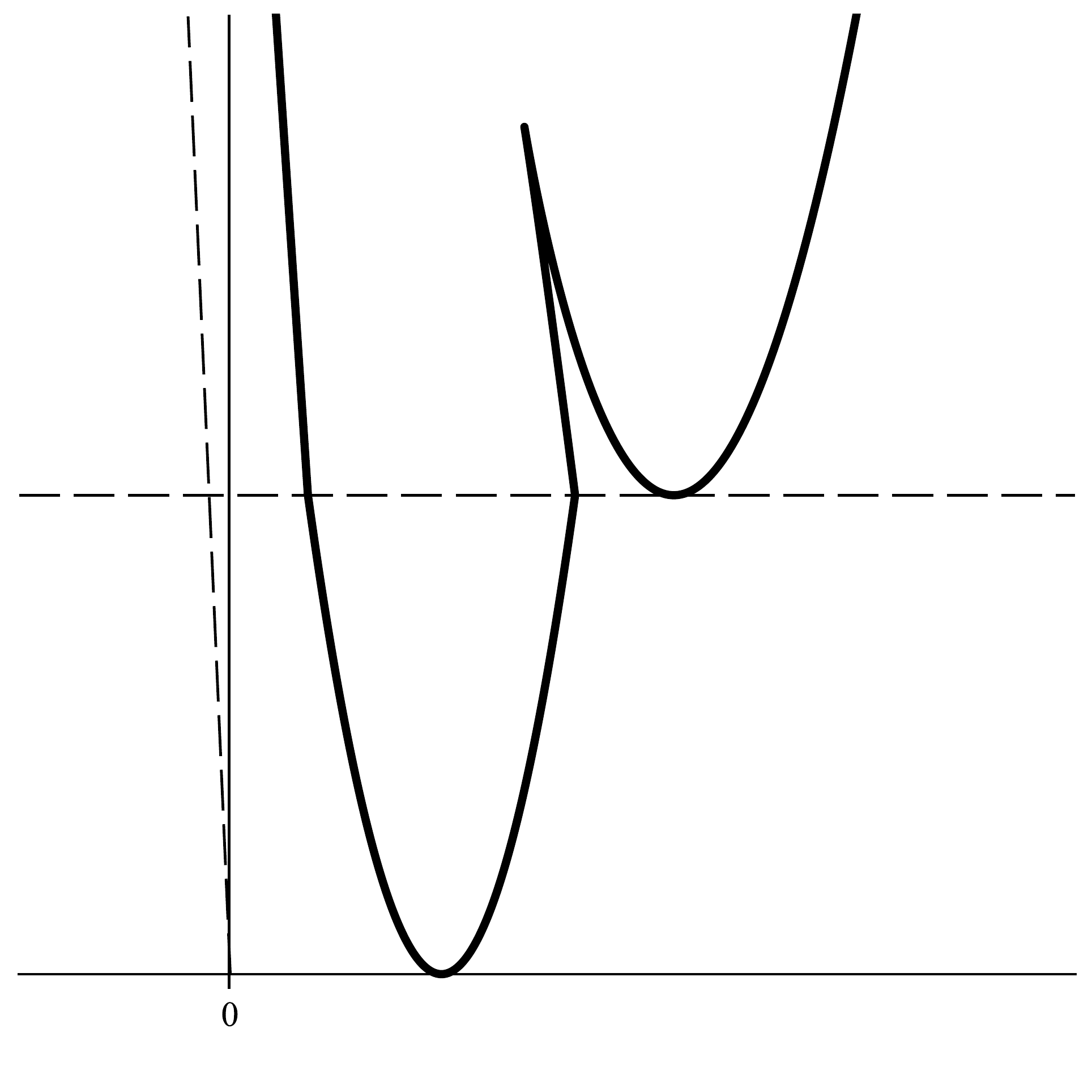}};
	\draw (2,0.4) node {\tiny{$m=m_0$}};
	\draw (-0.6,1.5) node {$\mathcal D$};

	\draw (-2.5,2) node {\tiny{$x+m=0$}};

	\end{tikzpicture}
		\caption{Two possible limiting domains $\mathcal D$, before (l) and after (r) the critical time. In the right figure, the particles in the upper part that are not blocked by the particles in  the lower part are drifting away from the lower part.  }\label{fig:domains}
}
\end{figure}

\begin{proposition}\label{prop0}
The map  $(\xi,\mu)\mapsto \Omega(\xi,\mu)$ is a homeomorphism from $\mathcal D$ to $\HH$. 
\end{proposition}


\begin{remark}\label{remark1}
As in \cite[eq (1.13)]{BoFe} we note that that map $\Omega:(\xi,\mu)\mapsto \Omega(\xi,\mu)$ solves the complex Burgers equation \cite{KO}
\begin{align}
\frac{\sigma}{(\sigma-\Omega)} \frac{\partial \Omega}{\partial \mu} = \frac{\partial \Omega}{\partial \xi},
\end{align}
where $\sigma=1$ if $\mu <\mu_0$ and $\sigma=2$ is $\mu\geq \mu_0$. The jump discontinuity in $\sigma$ is due to the different jump rates. 
\end{remark}

\begin{remark}
The boundary $\partial \mathcal  D$ is the set of all $(\xi,\mu)$ such that $\Omega(\xi,\mu)$ is a double critical point of $F$. For real values of $\xi$ and $\mu$, the only possible double critical points are real. This gives a way of explicitly computing the boundary. For $\Omega\in \R$, we solve the system $F'(\Omega)=F''(\Omega)=0$ for $(\xi,\mu)$. This method was used to draw the pictures in Figure~\ref{fig:domains}.  Moreover we have the following identification of points. The points $\Omega=0$ and $1$ correspond to the points $\partial \mathcal D$ touching the line $x+m=0$ and the horizontal axis. If we are beyond the critical time, we also have that $\Omega=2$ corresponds to the boundary point touching the line $m=m_0$ in the cloud starts drifting away. The cusp  corresponds to a point $\Omega\in (1,2)$. At the critical time (the birth of the cusp),  the cusp  and the lowest point in the cloud drifting away meet and correspond to $\Omega=2$.
\end{remark}
 In the following theorem we present our main result on the limit shape.
\begin{theorem}\label{th:macro}
Let $h$ be as defined in \eqref{eq:defh} and set 
\begin{align}\label{eq:thmacroscale}
	\left\{\begin{array}{l}
	t=[L \tau]\\
	m=[L\mu]\\
	m_0=[L \mu_0]\\ 
	x=[L \xi]\\
\end{array}\right.
\end{align}
Then  we have that 
\begin{align}
\lim_{L\to \infty} \frac{1}{L} \EE h(x,m)= \frac{1}{\pi} \Im F(\Omega(\xi,\mu))=:\overline{h}(\xi,\mu),
\end{align}
for $(\xi,\mu)\in \mathcal D$. \end{theorem}
The limiting height function is  a continuous function of the scaled variables. However, the normal to the surface constructed out of its graphs is not, as can be seen from the following result.
\begin{proposition}\label{prop:normal}
Let $n=(n_1,n_2,n_3)$ be the normal to the limiting surface defined by the graph of $\overline h$ defined in Theorem \ref{th:macro}. Normalize  $n$  such that $n_3=1$. Then we have
\begin{align}
\begin{pmatrix}
n_1\\
n_2\\
n_3
\end{pmatrix}
=\begin{pmatrix}
\theta_1/\pi\\
-\theta_3/\pi\\
1
\end{pmatrix},
\end{align} 
where $\theta_j$ are the  angles of the triangle formed by 
\begin{enumerate}
\item $0,1$ and $\Omega(\xi,\mu)$ if $\mu\leq \mu_0$,
\item $0,2$ and $\Omega(\xi,\mu)$ if $\mu>\mu_0$.
\end{enumerate}
See also Figure~\ref{fig:triangle}.
\end{proposition} 

\begin{figure}
\centering
\ifx\JPicScale\undefined\def\JPicScale{0.7}\fi
\unitlength \JPicScale mm
\begin{picture}(150,42)(0,30)
\linethickness{0.3mm}
\multiput(10,40)(0.12,0.24){83}{\line(0,1){0.24}}

\multiput(40,40)(6,0){5}{\line(1,0){3}}

\put(10,35){\makebox(0,0)[cc]{$0$}}
\put(40,35){\makebox(0,0)[cc]{$1$}}
\put(70,35){\makebox(0,0)[cc]{$2$}}

\put(14,42.5){\makebox(0,0)[cc]{$\theta_1$}}
\put(34,42.5){\makebox(0,0)[cc]{$\theta_2$}}
\put(21,56){\makebox(0,0)[cc]{$\theta_3$}}

\put(94,42.5){\makebox(0,0)[cc]{$\theta_1$}}
\put(137,42.5){\makebox(0,0)[cc]{$\theta_2$}}
\put(102,56){\makebox(0,0)[cc]{$\theta_3$}}

\put(20,65){\makebox(0,0)[cc]{$\Omega$}}
\put(90,35){\makebox(0,0)[cc]{$0$}}
\put(120,35){\makebox(0,0)[cc]{$1$}}
\put(150,35){\makebox(0,0)[cc]{$2$}}
\put(100,65){\makebox(0,0)[cc]{$\Omega$}}

\put(10,40){\circle*{1.5}}
\put(40,40){\circle*{1.5}}
\put(70,40){\circle*{1.5}}
\put(20,60){\circle*{1.5}}

\put(90,40){\circle*{1.5}}
\put(120,40){\circle*{1.5}}
\put(150,40){\circle*{1.5}}
\put(100,60){\circle*{1.5}}

\linethickness{0.3mm}
\multiput(20,60)(0.12,-0.12){167}{\line(1,0){0.12}}

\linethickness{0.3mm}
\put(10,40){\line(1,0){30}}

\linethickness{0.3mm}
\multiput(90,40)(0.12,0.24){83}{\line(0,1){0.24}}

\linethickness{0.3mm}
\multiput(100,60)(0.3,-0.12){167}{\line(1,0){0.3}}

\linethickness{0.3mm}
\put(90,40){\line(1,0){60}}

\end{picture}
\caption{The angles $\theta_j$ in Proposition \ref{prop:normal}} \label{fig:triangle}
\end{figure}
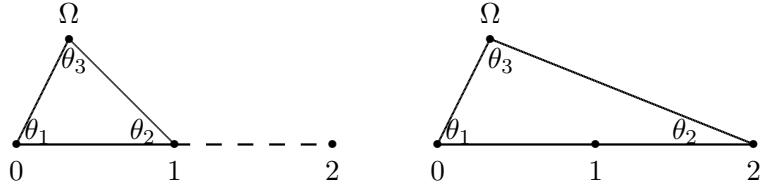
The discontinuity in the limit shape can also be seen from the PDE that it satisfies. See also Remark \ref{remark1}.
\begin{proposition}\label{prop:PDE}
The limiting mean height function $\overline{h}$ from Theorem \ref{th:macro} satisfies the PDE
\begin{align}
\frac{{\partial} \overline{h}}{{\partial }\tau}=\begin{cases}
f(\nabla \overline{h}),& \mu<\mu_0 \\
2f(\nabla \overline{h}),& \mu>\mu_0 \end{cases}
\end{align}
where
\begin{align}
f(x,y)=-\frac{\sin \pi x  \sin\pi(y-x)}{\pi \sin \pi y}.
\end{align}
\end{proposition}
One easily computes that the signature of the Hessian of $f$ is $(1,-1$).  This implies that our model is locally in the  class of growth models described by the anisotropic KPZ equation \cite{BS} as mentioned in the Introduction.  However, the PDE has a jump discontinuity at the separating line.

\subsection{Gaussian free field}

The main result of the paper is on the fluctuation of $h$ around $\EE h$, which turn out to be governed by the Gaussian free field. For a survey on the Gaussian free field see \cite{S}. 

Let us first  recall some definitions. Denote  the Laplace operator on $\HH$ with Dirichlet boundary conditions by $\triangle$  and consider the Sobolev space $\mathcal W_0$ on $\HH$ defined as the completion of the space of smooth functions with compact support in $\HH$ equipped with inner product\begin{align}\label{eq:sobolevinpr}
(\phi_1,\phi_2)_\nabla=\int_\HH \nabla \phi_1 \cdot \nabla\phi_2.\end{align}
Note that by integrating by parts and using the fact we have Dirichlet boundary conditions we have
\[(\phi_1,\phi_2)_\nabla=-\int_\HH \phi_1 \triangle \phi_2,\]
for sufficiently smooth $\phi_2$. 

The Gaussian free field  on $\HH$ is a collection of centered Gaussian  random variables $\{\langle F,\phi\rangle_\nabla \}_{\phi\in \mathcal W_0}$ (indexed by the Sobolev space $\mathcal W_0$) and covariance 
\begin{align}\label{eq:GFFcovariance}
\EE \left[\langle F,\phi_1\rangle_\nabla \langle F,\phi_2\rangle_\nabla\right] =(\phi_1,\phi_2)_\nabla.
\end{align}

The main result of the paper is that  the fluctuations $h-\EE h$ in the large $L$ limit are governed by the Gaussian free field. More precisely, the push forward of $h-\EE h$ under the map $(\xi,\mu)\mapsto \Omega(\xi,\mu)$ converges to the Gaussian free field. The proof that we present here, differs from the usual approach to compute the moments of the height fluctuations at different points. Instead, inspired by \eqref{eq:GFFcovariance} we define the pairing of the height function with a test function and compute the limiting behavior of the characteristic function.  Care should be taken here, since the function $h-\EE h$ is defined on a discrete set. We therefore restrict ourselves  to test functions $\phi$ that are $C^2$ and have compact support in ${\HH}$. Then we define the pairing between $h-\EE h$ and $\phi$ by a discretization of the Sobolev inner product \eqref{eq:sobolevinpr}  and including the pushforward with $\Omega$.

For a function $F:\Z\times \N\to \R$ and a $C^2$ function $\phi$ with compact support in  $\HH$, define the pairing $\langle F, \phi\rangle$ by 
\begin{align} \label{eq:discretepairing}
\langle F,\phi\rangle= -\frac{\sqrt{\pi}}{ L^2}\sum_{(x,m)\in L \mathcal D}F(x,m) \triangle\phi\left(\Omega(x/L,m/L)\right) J(x/L,m/L),
\end{align}
where $J$ stands for the Jacobian of the map $(\xi,\mu)\mapsto \Omega(\xi,\mu)$, i.e.
\begin{align}\label{eq:Jacobian}
J(\xi,\mu)=\left|\det \begin{pmatrix}
\frac{\partial \Re \Omega}{\partial \xi} &\frac{\partial \Re \Omega}{\partial \mu}\\
\frac{\partial \Im \Omega}{\partial \xi}& \frac{\partial \Im \Omega}{\partial \mu}
\end{pmatrix} \right|
\end{align}

The main result of this paper is formulated in the following theorem.
\begin{theorem}\label{th:main}
Let $\phi$ be a $C^2$ function with compact support on $\HH$ and $h$ be the height function defined in \eqref{eq:defh}. Then with $\langle \cdot, \cdot \rangle$ as in \eqref{eq:discretepairing} we have
\begin{equation}\label{eq:thmain}
\lim_{L\to \infty} \EE [\exp {\rm i} t \langle h-\EE h,\phi\rangle]={\rm e}^{-\frac{\|\phi\|_\nabla^2 t^2 }{2}} .
\end{equation}
Hence, as $L\to \infty$, the pushfoward of the random surface defined by $h-\EE h$ by the map $\Omega:\mathcal D\to \HH$ converges to the Gaussian free field on $\HH$.
\end{theorem}

The key fact that we use to prove Theorem \ref{th:main}  is that the process on the interlacing particles at time $t$ defines a determinantal point process, as we will discuss in the next paragraph.

\begin{remark}
The fact that \eqref{eq:thmain} implies \eqref{eq:GFFcovariance} follows by the polarization identity for norms in Hilbert spaces and the fact that  the pairing in \eqref{eq:discretepairing} is linear. See also \cite[Prop. 2.13]{S}.
\end{remark}

\begin{remark}
Note that in Theorem \ref{th:macro}  the height function grows with $L$ as $L\to\infty$, whereas the fluctuations in Theorem \ref{th:main} are not scaled with $L$ at all.  However, the Gaussian free field is a probability measure on generalized functions. The pointwise limit in the scaled variables of $h-\EE h$ is ill-defined. Indeed, the variance  at a given point grows logarithmically with $L$. In contrast, the correlation between  two separate points remains bounded and is expressed in terms of the Green's function for the Laplace operator on $\HH$ with Dirichlet boundary conditions. See also~\cite{S}. \end{remark}

\subsection{Determinantal point processes}

For a discrete set $\mathcal X$ a determinantal point process on $\mathcal X$ is a probability measure on $2^\mathcal X$ for which there exists a kernel $K:\mathcal X\times \mathcal X\to \C$
such that 
\begin{align}\label{eq:dpp}
	\mathrm{Prob}\{X \in 2^\mathcal X  \mid Y\subset X)=\det K(x,y)_{x,y\in Y}
\end{align}
for all finite sets $Y\subset X$. For more details on determinantal point processes we refer to \cite{BorDet,HKPV,J,K,L,Sosh,Sosh2}. A determinantal point process is completely determined by its kernel.


The fact of the matter is that the growth model on the interlacing particles system at time  $t$, defines a determinantal point process on  $\Z\times \N$
with kernel given by (see \cite{BoFe}) 
\begin{multline}\label{eq:defKdoubleint}
	K(x_1,m_1, x_2,m_2) =-\frac{\chi_{m_1< m_2}}{2\pi{\rm i}} \oint_{\Gamma_0}    \frac{p_{m_1}(w)}{p_{m_2}(w)}  \frac{{\rm d} w}{w^{x_1+m_1-x_2-m_2+1}}\\
+\frac{1}{(2\pi{\rm i})^2} \oint_{\Gamma_0} \oint_{\Gamma_{1,2}}
\frac{{\rm e}^{tw}  p_{m_1}(w) z^{x_2+m_2}}
	{{\rm e}^{tz} p_{m_2}(z) w^{x_1+m_1}}\frac{ {\rm d}z   {\rm d}w}{w(w-z)}.
\end{multline}
Here $p_m$ is defined by
\begin{align}\label{eq:defpm}
	p_m(z)=\left\{\begin{array}{cc}
	(z-1)^m, & m \leq m_0\\
	(z-1)^{m_0}(z-2)^{m-m_0}, & m\geq m_0.
\end{array}\right.\end{align}
Moreover,  $\Gamma_0$ is a contour that encircles the pole $w=0$ but no other and is equipped  with counterclockwise orientation. The contour $\Gamma_{1,2}$ encircles the poles $z=1,2$ but no other and is also equipped with counterclockwise orientation.
 \begin{remark}
The expression \eqref{eq:defKdoubleint}  for the kernel  can be obtained from the expression as given in \cite{BoFe} by the transform $w\mapsto 1/w$ and $z\mapsto 1/z$.
\end{remark}

The fact that we deal with an determinantal process with an explicitly known kernel, allows us to analyze the system in detail. The main idea behind the proof of Theorem \ref{th:main} is to write $\langle h-\EE h,\phi\rangle$ as a linear statistic, which allows us to rewrite the characteristic function at the left-hand side of \eqref{eq:thmain}  as a (Fredholm) determinant. The proof of Theorem \ref{th:main} is then given in two steps: first we compute the limiting behavior of the variance of $\langle h-\EE h, \phi\rangle$ and then we use the Fredholm determinant representation for the characteristic function to prove that the fluctuations are Gaussian.

 As we will show (see Proposition \ref{prop:variancelinstat}), the variance of $\langle h-\EE h, \phi\rangle$ can be expressed as a double sum involving the kernel $K$ in \eqref{eq:defKdoubleint}.    The asymptotic behavior for $L\to \infty$ can then be found by an asymptotic analysis for the double integral representation for $K$  based on saddle point methods. However, the disadvantage of this approach is that we also need to control the asymptotic behavior of the $K$ near the boundary points, which then afterwards turn out to be negligible (note that the Gaussian free field is defined on a Sobolev space with Dirichlet boundary conditions).  For this reason, we will follow a different approach that  shows that the boundary is redundant  in a more direct way. Instead of computing the asymptotics for $K$ first, we provide an alternative expression for the double sum in the variance in terms of a quadruple integral  and compute the (bulk) asymptotics afterwards. Moreover, based on the fact that we deal with  Dirichlet boundary conditions, we give a probabilistic argument to show that the boundary has no effect on the fluctuations and that we can ignore the boundary in the Fredholm determinant identity for the characteristic function of  $\langle h-\EE h,\phi\rangle$.  Apart from the conceptual benefit, this approach also reduces the amount of  computational technicalities. Indeed,  as the boundary has  a complicated structure due to the different jump rates, giving case by case estimates on the kernel near the different parts of the boundary is cumbersome.

\subsection{Overview of the rest of the paper}
The rest of the paper is organized as follows. In Section \ref{sec:newco} we prove Propositions \ref{prop0}, \ref{prop:normal} and \ref{prop:PDE} on the complex structure. In Section \ref{sec:var} we compute the limiting behavior of the variance of $\langle h-\EE h,\phi\rangle$. In Section \ref{sec:proof} we prove that the fluctuations are Gaussian and by combining this with the results of Section \ref{sec:var} we obtain a proof for Theorem \ref{th:main}.  Some of the statements in Sections \ref{sec:var} and \ref{sec:proof} are based on steepest descent arguments that we postpone to Sections \ref{sec:asymptoticsK} and \ref{sec:asymptoticsR}.  More precisely, Proposition \ref{prop:asymptR1} is proved in \ref{sec:asymptoticsR} and Lemma \ref{lem:conj} is proved in Section \ref{sec:asymptoticsK}.
The proof of Theorem \ref{th:macro} is, although standard, given for completeness in Section \ref{sec:asymptoticsK}.

\section{Proof of Propositions \ref{prop0}, \ref{prop:normal} and \ref{prop:PDE}} \label{sec:newco}
In this section we prove Propositions \ref{prop0}, \ref{prop:normal} and \ref{prop:PDE}.

\begin{proof}[Proof of Propsition \ref{prop0}]
We  prove  that to every $\Omega\in \HH$ there exists a unique $(\xi,\mu)\in \mathcal D$ such that $\Omega=\Omega(\xi,\mu)$.  Let us first consider \eqref{eq:defF} for $\mu\leq \mu_0$. By taking real and imaginary parts we can put  $F'(\Omega \mid \xi,\mu)=0$  in matrix form as
\begin{align}\label{eq:matrixeq1}
\begin{pmatrix}
\Re \frac{1}{\Omega} & \Re\left(\frac{1}{\Omega}-\frac{1}{\Omega-1} \right) \\
\Im \frac{1}{\Omega}& \Im \left( \frac{1}{\Omega}-\frac{1}{\Omega-1}\right)
\end{pmatrix}
\begin{pmatrix}
\xi\\
\mu
\end{pmatrix}=\begin{pmatrix} \tau\\
0\end{pmatrix}.
\end{align}
The determinant of this matrix can be easily computed to be
 \[\det\begin{pmatrix}
\Re \frac{1}{\Omega} & \Re\left(\frac{1}{\Omega}-\frac{1}{\Omega-1} \right) \\
\Im \frac{1}{\Omega}& \Im \left( \frac{1}{\Omega}-\frac{1}{\Omega-1}\right)
\end{pmatrix}=\frac{\Im \Omega}{|\Omega(\Omega-1)|^2},
\]
which is always strictly positive for $\Omega \in \HH$. Hence we can invert the matrix and solve \eqref{eq:matrixeq1}
\begin{align}\label{eq:inverse1}\begin{pmatrix}
\xi\\
\mu
\end{pmatrix}
=\tau \begin{pmatrix}
|\Omega|^2-|\Omega-1|^2\\
|\Omega-1|^2
\end{pmatrix}.\end{align}
We recall that this formula is only valid for $\mu\leq \mu_0$. Concluding, for any $\Omega \in \HH$ with $|\Omega-1|^2\leq \mu_0/\tau$ there is  a unique pair $(\xi,\mu)$ such that $\Omega=\Omega(\xi,\mu)$. 

By repeating the same procedure  but now for $\mu>\mu_0$ we  obtain
\begin{align}\label{eq:inverse2}
\begin{pmatrix}
\xi\\
\mu
\end{pmatrix}
=
\begin{pmatrix}
\frac{|\Omega|^2}{2}\left(\tau +\frac{\mu_0}{|\Omega-1|^2}\right)
-\mu\\
\mu_0+\frac{|\Omega-2|^2}{2}\left(\tau -\frac{\mu_0}{|\Omega-1|^2}\right) \end{pmatrix}.\end{align}
Of course, this is only valid for $\mu> \mu_0$ so that we need $\tau -\frac{\mu_0}{|\Omega-1|^2}>0$. Hence  for every $\Omega\in \HH$ with  $|\Omega-1|^2> \mu_0/\tau$ there is a unique pair $(\xi,\mu)$ such that $\Omega=\Omega(\xi,\mu)$. 

Concluding, we find that the map $(\xi,\mu)\mapsto \Omega(\xi,\mu)$  is a one-to-one map from $ \{(\xi,\mu)\in \mathcal D \mid \mu\leq \mu_0\}$ onto
\[  \{\Omega \in \HH  \mid |\Omega-1|^2\leq \frac{\mu_0}{\tau}\},\]
and restricted to $ \{(\xi,\mu)\in \mathcal D \mid \mu> \mu_0\}$  it is one-to-one
to 
\[  \{\Omega\in \HH  \mid |\Omega-1|^2> \frac{\mu_0}{\tau}\},\]
Since these sets are disjoint and the union is $\HH$, we see that the unrestricted  map $(\xi,\mu)\mapsto \Omega(\xi,\mu)$ is indeed one-to-one from $\mathcal D$ onto $\HH$. 

The continuity is statement is immediate.
\end{proof}
\begin{proof}[Proof of Proposition \ref{prop:normal}]
The vector $(-{\partial \overline{h} }/{\partial }\xi ,-{\partial \overline h}/{\partial }\mu,1)$ is a normal to the graph of $\overline{h}$.  By Theorem \ref{th:macro}, taking  total derivatives and using $F'(\Omega)=0$, we obtain 
\begin{align}\label{eq:partialxi}
\frac{{\partial }}{{\partial }\xi} \overline{h}=\frac{1}{\pi} \frac{{\partial }}{{\partial }\xi} \Im F(\Omega)=-\frac{1}{\pi} \arg \Omega,
\end{align}
and
\begin{align}\label{eq:partialmu}
\frac{{\partial  }}{{\partial }\mu} \overline{h}=\frac{1}{\pi} \frac{{\partial}}{{\partial }\mu} \Im F(\Omega)=\frac{1}{\pi} \arg(\Omega -\sigma)-\frac{1}{\pi} \arg \Omega,
\end{align}
where $\sigma=1$  if $\mu\leq \mu_0$ and $\sigma=2$ if $\mu>\mu_0$.  Hence by definition of $\theta_j$ we have
\begin{align}
-\frac{\partial }{\partial \xi} \overline{h}=\frac{\theta_1}{\pi}, \quad \text{and} \quad -\frac{{\partial}}{{\partial}\mu} \overline{h}=\frac{\theta_1}{\pi}- \frac{\pi-\theta_2}{\pi}=-\frac{\theta_3}{\pi},\end{align}
which proves the statement.
\end{proof}

\begin{proof}[Proof of Proposition \ref{prop:PDE}]
By arguing as in the proof of Proposition \ref{prop:normal} above, we obtain 
\begin{align}\label{eq:partialtau}
\frac{\partial }{{\partial} \tau} \overline{h}=\frac{1}{\pi} \frac{\partial }{{\partial} \tau} \Im F(\Omega)=\frac{\Im \Omega}{\pi}. 
\end{align} 
By the law of sines we have
\begin{align}
\frac{\sin \theta_2}{|\Omega|}=\frac{\sin \theta_3}{\sigma}.
\end{align}
And hence we can write
\begin{align}
\Im \Omega=|\Omega| \sin \theta_1=\frac {\sin \theta_1 \sin \theta_2}{\sigma \sin \theta_3}=\frac{\sin \theta_1 \sin (\theta_1+\theta_3)}{\sigma \sin \theta_3},
\end{align}
where $\sigma=1$  if $\mu\leq \mu_0$ and $\sigma=2$ if $\mu>\mu_0$. The statement now follows by  writing  $\Im \Omega$ and $\theta_j$ in terms of the partial derivatives of $\overline{h}$ with respect to $\xi,\mu$ and $\tau$ using \eqref{eq:partialxi}, \eqref{eq:partialmu} and \eqref{eq:partialtau}. 
\end{proof}

\section{The variance of $\langle h-\EE h,\phi\rangle$} \label{sec:var}

The purpose of this section is to compute the limit of the variance of $\langle h-\EE h,\phi\rangle$ as $L\to \infty$. 
\begin{proposition}\label{prop:variance} Let $h$ be as in \eqref{eq:defh}, $\phi$ a $C^2$ function with compact support in $\HH$ and $\langle\cdot,\cdot\rangle$ as in \eqref{eq:discretepairing}. Then 
\begin{align}\label{eq:propvariance}
\lim_{L\to \infty} \var  \langle h-\EE h,\phi\rangle = \|\phi\|_{\nabla}^2,
\end{align}
where $\|\cdot\|_{\nabla}$ is the Sobolev norm \eqref{eq:sobolevinpr}.
\end{proposition}
The proof will be given in Section 4.2. We will first provide a formula that expresses the variance of a general  linear statistic in terms of a kernel $R$.   We will show that $\langle h-\EE h,\phi\rangle$ is in fact a linear statistic and Proposition~\ref{prop:variance} then follows by  the asymptotic behavior for $R$ as $L\to \infty$. The proof of the asymptotic behavior of $R$ is based on a steepest descent argument and will be postponed to Section  \ref{sec:asymptoticsR}. 
\subsection{The variance of a linear statistic}
Let $f:\Z\times \N \to \R$ be a function with finite support and define the random variable $X_f$ as 
\begin{align}\label{eq:linstat}
X_f=\sum_{(x,m)\in \mathcal C} f(x,m),
\end{align}
where $\mathcal C$ is a random configuration of points.

Random variables of the type $X_f$ are generally referred to as linear statistics and  play an important role in the study of determinantal point processes.  It is standard (and straightforward to derive from \eqref{eq:dpp}) that 
\begin{align}\label{eq:meanoflinstat}
\EE X_f=\sum_{(x,m)} f(x,m) K(x,m,x,m),
\end{align}
and
\begin{multline}\label{eq:varianceoflinstat}
\var X_f=\sum_{(x,m)\in \Z\times \N} f(x,m)^2K(x,m,x,m)\\
-\sum_{(x_1,m_1)\in \Z\times \N }\sum_{(x_2,m_2)\in \Z\times \N} f(x_1,m_1) f(x_2,m_2) K(x_1,m_1,x_2,m_2) K(x_2,m_2,x_1,m_1), 
 \end{multline}
 where $K$ is the kernel in \eqref{eq:defKdoubleint}.
 We will rewrite \eqref{eq:varianceoflinstat} in terms of a kernel $R$. To this end, we need the following lemma. 
 
 \begin{lemma}\label{lem:KsquareK} With $K$ as in \eqref{eq:defKdoubleint} we
 we have that 
	\begin{align} \label{eq:KsquareK}
	  \sum_{x_2\in \Z} K(x_1,m_1,x_2,m_2)K(x_2,m_2,x_1,m_1)=  
	  K(x_1,m_1; x_1,m_1) \delta_{m_1,m_2}	\end{align}
for $(x_1,m_1)\in \mathcal \Z\times \N$ and $m_2\in \N$.
\end{lemma}
\begin{proof}
We start by  rewriting the integral formulas for $K$ in \eqref{eq:defKdoubleint} as follows. If $m_1\geq m_2$, we deform $\Gamma_0$ such that it also goes around $\Gamma_{1,2}$. Due to the term $(w-z)^{-1}$ we pick up a residue that results into single integral over $\Gamma_{1,2}$. However, by the assumption $m_1\geq m_2$ the integrand of the single integral has no pole inside  $\Gamma_{1,2}$.  Hence it vanishes and we are left with the double integral only, where now $\Gamma_0$ also goes around $\Gamma_{1,2}$.  If on the other hand, $m_1<m_2$, then we deform the contour $\Gamma_{1,2}$ such that it also goes around $\Gamma_0$. The residue that we pick up in this way results into a single integral over $\Gamma_0$ that exactly cancels the single integral that was already present in \eqref{eq:defKdoubleint}. Concluding we can rewrite \eqref{eq:defKdoubleint} as 
\begin{align}\label{eq:KsquareK1}
	K(x_1,m_1, x_2,m_2) =\frac{1}{(2\pi{\rm i})^2} \oint_{\Gamma_0}\oint_{\Gamma_{1,2}}
\frac{{\rm e}^{tw}  p_{m_1}(w) z^{x_2+m_2}}
	{{\rm e}^{tz} p_{m_2}(z) w^{x_1+m_1}}\frac{   {\rm d}z  {\rm d}w }{w(w-z)},
\end{align}
where $\Gamma_0$ now also goes around $\Gamma_{1,2}$ if $m_1<m_2$, and if on the other hand $m_1\geq m_2$, we have that $\Gamma_{1,2}$ also goes around $\Gamma_0$.  See also Figure~\ref{fig:m1m2}.

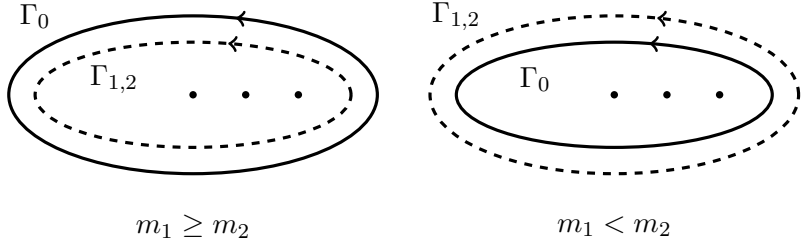
\begin{figure}[t]
\begin{center}
\begin{tikzpicture}[scale=0.7]
\fill (0,0) circle(2pt);
\fill (1,0) circle(2pt);
\fill (2,0) circle (2pt);
\draw[postaction={decorate},decoration={
  markings,
  mark=at position 0.2 with {\arrow{>}}},very thick,dashed] (0,0) ellipse (3 and 1);
\draw[postaction={decorate},decoration={
  markings,
  mark=at position 0.2 with {\arrow{>}}},very thick] (0,0) ellipse (3.5 and 1.5);
\fill (8,0) circle(2pt);
\fill (9,0) circle(2pt);
\fill (10,0) circle (2pt);
\draw (0,-2.5) node  {$m_1\geq m_2$};
\draw (8,-2.5) node  {$m_1<m_2$};
\draw (-3,1.5) node {$\Gamma_0$};
\draw (-1.5,0.3) node {$\Gamma_{1,2}$};
\draw (5,1.5) node {$\Gamma_{1,2}$};
\draw (6.5,0.3) node {$\Gamma_{0}$};

\draw[postaction={decorate},decoration={
  markings,
  mark=at position 0.2 with {\arrow{>}}},very thick] (8,0) ellipse (3 and 1);
\draw[postaction={decorate},decoration={
  markings,
  mark=at position 0.2 with {\arrow{>}}},very thick,dashed]  (8,0) ellipse (3.5 and 1.5);
\end{tikzpicture}
\end{center}
\caption{Deforming the contours such that the kernel $K(x_1,m_1,x_2,m_2)$ is represented by one double contour only. The solid line represents the contour $\Gamma_0$. The dashed contour represents $\Gamma_{1,2}$. The dots are the points $0,1$ and $2$ which are  poles for the integrands.}
\label{fig:m1m2}
\end{figure}
The next step in proving \eqref{eq:KsquareK} is to write $K(x_1,m_1,x_2,m_2)$ and 
$K(x_2,m_2,x_1,m_1)$ both as one double integral. Note that we can always deform the contours such that all contours are circles around the origin and such that the radius of $\Gamma_0'$ and $\Gamma_{1,2}$ are equal.  After these preparations we substitute \eqref{eq:defKdoubleint} into the sum \eqref{eq:KsquareK} and take the sum under the integral. The sum is then over terms $\frac{z^{x_2}}{w'^{x_2}}$ and we have
\begin{multline}\label{eq:KsquareKsum} \sum_{x_2\in \Z} \frac{1}{(2\pi {\rm i})^2} \oint\oint  \frac{{\rm e}^{tw'}p_{m_2}(w')}{{\rm e}^{tz} p_{m_{2}}(z)} \frac{z^{x_2+m_2}}{w^{x_2+m_2+1}} \frac{ {\rm d}w' {\rm d}z}{(w-z)(w'-z')}\\=\frac{1}{2\pi{\rm i} } \oint_{\Gamma_{1,2}} \frac{{\rm d}z}{(w-z)(z-z')}.
\end{multline}
The integral at the right-hand side vanishes if $w$ and $z'$ are  either both inside the contour $\Gamma_{1,2}$ or both outside. These situations precisely occur when $m_1\neq m_2$. Hence we  proved \eqref{eq:KsquareK} for $m_1\neq m_2$. Finally, if $m_1=m_2$, we have that $z'$ is in the region enclosed by $\Gamma_{1,2}$ and hence the left-hand side of \eqref{eq:KsquareK}, by using \eqref{eq:KsquareKsum} and picking up the residue at $z=z'$, equals
\begin{align}
\frac{1}{(2\pi{\rm i})^2} \oint_{\Gamma_0} \oint_{\Gamma_{1,2'}}
\frac{{\rm e}^{tw}  p_{m_1}(w) z'^{x_1+m_1}}
	{{\rm e}^{tz'} p_{m_1}(z') w^{x_1+m_1}}\frac{ {\rm d}z'   {\rm d}w}{w(w-z')},
\end{align}
where $\Gamma_0$ goes around $\Gamma_{1,2'}$. Hence it equals $K(x_1,m_1,x_1,m_1)$ by \eqref{eq:KsquareK1} and we also proved \eqref{eq:KsquareK} for the remaining case $m_1=m_2$.\end{proof}

Using the fact that the diagonal of $K^2$ and $K$ agree,  we can rewrite the variance of  linear statistic in a useful different form  as presented in the following proposition.
\begin{proposition} \label{prop:variancelinstat}
Define the difference operator $D$ by  $Df(x,m)=f(x,m)-f(x-1,m)$ for $f:\Z\times \N\to \R$. Then 
\begin{multline}
\var X_f= \sum_{(x_1,m_1)\in \Z\times \N } \sum_{(x_2,m_2)\in \Z\times \N}Df(x_1,m_1) Df(x_2,m_2) R(x_1,m_1,x_2,m_2),
\end{multline}
where 
\begin{multline}\label{eq:defR}
R(y_1,m_1,y_2,m_2)\\=
\left\{
\begin{array}{ll}
\sum_{x_1\geq y_1}\sum_{x_2< y_2} K(x_1,m_1,x_2,m_2)K(x_2,m_2,x_1,m_1)  & y_1\geq y_2,\\
\sum_{x_1< y_1}\sum_{x_2\geq  y_2} K(x_1,m_1,x_2,m_2)K(x_2,m_2,x_1,m_1)  & y_1<y_2.
\end{array}
\right.
\end{multline}
\end{proposition}
\begin{proof}
The right-hand side of \eqref{eq:varianceoflinstat} consists of two terms.  Let us first consider the first term.  By writing $f(x_j,m_j)=\sum_{y_j\leq x_j} D f$ and changing the order of summation we get
\begin{multline*}
\sum_{x_1,m_1} f(x_1,m_1)^2K (x_1,m_1,x_1,m_1)\\
=\sum_{y_1,m_1}\sum_{y_2} Df(y_1,m_1) Df(y_2,m_1) \sum_{x_1\geq \max(y_1,y_2)}K(x_1,m_1,x_1,m_1).
\end{multline*}
By separating the cases $y_1<y_2$ and $y_1\geq y_2$ and using \eqref{eq:KsquareK} to introduce  extra sums over $x_2$ and $m_2$,  we obtain 
\begin{multline}\label{eq:expandtrace2}
\sum_{x_1,m_1} f(x_1,m_1)^2K (x_1,m_1,x_1,m_1)=\sum_{y_1,m_1}\sum_{y_2,m_2} Df(y_1,m_1) Df(y_2,m_2)\\\times 
\left\{
\begin{array}{ll}
\sum_{x_1\geq y_1}\sum_{x_2} K(x_1,m_1,x_2,m_2)K(x_2,m_2,x_1,m_1),  & y_1\geq y_2\\
\sum_{x_1}\sum_{x_2\geq y_2}K(x_1,m_1,x_2,m_2)  K(x_2,m_2,x_1,m_1), & y_1<y_2.\\
\end{array}
\right.
\end{multline}
 Note that after summing over $x_2$, the summand vanishes if $m_1\neq m_2$.
 
 The second term in \eqref{eq:varianceoflinstat} is a double sum that we can rewrite as
\begin{multline*}
\sum_{x_1,m_1}\sum_{x_2,m_2} f(x_1,m_1)f(x_2,m_2) K(x_1,m_1,x_2,m_2)K(x_2,m_2,x_1,m_1)\\
=\sum_{x_1,m_1}\sum_{x_2,m_2}\sum_{y_1\leq x_1}D f(y_1,m_1) \sum_{y_2\leq x_2}Df(y_2,m_2)\\ \times K(x_1,m_1,x_2,m_2)K(x_2,m_2,x_1,m_1),
\end{multline*}
which after changing the order of summation reduces to
\begin{multline}\label{eq:expandtrace1}
\sum_{x_1,m_1}\sum_{x_2,m_2} f(x_1,m_1)f(x_2,m_2) K(x_1,m_1,x_2,m_2)K(x_2,m_2,x_1,m_1)\\=\sum_{y_1,m_1}\sum_{y_2,m_2} Df(y_1,m_1) Df(y_2,m_2)\\\times 
\sum_{x_1\geq y_1}\sum_{x_2 \geq y_2} K(x_1,m_1,x_2,m_2)K(x_2,m_2,x_1,m_1).
\end{multline}
Inserting \eqref{eq:expandtrace1} and \eqref{eq:expandtrace2} in  \eqref{eq:varianceoflinstat} gives the statement.
\end{proof}
We end this paragraph with a symmetry property of $R$ that is useful later on.
\begin{lemma} \label{prop:symR}
If $m_1\neq m_2$ then \begin{multline}\label{eq:symR}
\sum_{x_1\geq y_1}\sum_{x_2< y_2} K(x_1,m_1,x_2,m_2)K(x_2,m_2,x_1,m_1) \\
=\sum_{x_1< y_1}\sum_{x_2\geq  y_2} K(x_1,m_1,x_2,m_2)K(x_2,m_2,x_1,m_1)  
\end{multline}
\end{lemma}
\begin{proof} 
The proof follows by applying \eqref{eq:KsquareK} twice and changing the order of summation
\begin{multline}
\sum_{x_1\geq y_1}\sum_{x_2< y_2} K(x_1,m_1,x_2,m_2)K(x_2,m_2,x_1,m_1) \\
=-\sum_{x_1\geq y_1}\sum_{x_2\geq  y_2} K(x_1,m_1,x_2,m_2)K(x_2,m_2,x_1,m_1)  
\\=\sum_{x_1< y_1}\sum_{x_2\geq y_2} K(x_1,m_1,x_2,m_2)K(x_2,m_2,x_1,m_1).\end{multline}
\end{proof}

\subsection{Proof of Proposition \ref{prop:variance} }In this paragraph we prove Proposition \ref{prop:variance}.
We need the following lemma that expresses  $\langle h, \phi \rangle $ as a linear statistic.
\begin{lemma}\label{lem:linearstat}
{Let $h$ be as in \eqref{eq:defh}, $\phi$ a $C^2$ function with compact support in $\HH$, $\langle h, \phi\rangle$ as in \eqref{eq:discretepairing} and $f$ be defined by \begin{align}\label{eq:deff}
f(x,m)=-\frac{\sqrt{\pi}}{L^2}\sum_{\overset{y\leq x}{(y,m)\in L\mathcal D}} \triangle\phi(\Omega(y/L,m/L)) J(y/L,m/L).\end{align} 
Then $\langle h,\phi\rangle=X_f$.} \end{lemma}
\begin{proof}
Let $\mathcal C$ be a random configuration of points. Then
\[\langle h, \phi\rangle=-\frac{\sqrt{\pi}}{L^2} \sum_{(y,m)\in L \mathcal D} h(y,m) \triangle \phi(\Omega(y/L,m/L))J(y/L,m/L).\]
By inserting definition of $h$ in \eqref{eq:defh} and changing the order of summation we get
\[\langle h, \phi\rangle=-\frac{\sqrt{\pi}}{L^2} \sum_{(x,m)\in \mathcal C} \sum_{\overset{y\leq x}{(y,m)\in L\mathcal D}} \triangle \phi(\Omega(y/L,m/L)) J(y/L,m/L),\]
which is the statement. 
\end{proof}
Now that we have shown that $\langle h,\phi\rangle$ is a linear statistic, we can express its variance in terms of $R$ by Proposition \ref{prop:variancelinstat}. Hence it suffices to find the asymptotic behavior of $R$, that we present in the following proposition.

\begin{proposition} \label{prop:asymptR1}
Fix $\delta>0$ and set 
\begin{align}
\begin{cases} x_j= [L \xi_j],\\
m_j=[L \mu_j],
\end{cases}
\end{align}
where $[x]$ stands for the integer part of $x$. Then
\begin{align}
\label{eq:limR1} \lim_{L\to \infty} R(x_1,m_1,x_2,m_2)=-\frac{1}{2 \pi^2} \log\left|\frac{\Omega(\xi_1,\mu_1)-\Omega(\xi_2,\mu_2)}{\Omega(\xi_1,\mu_1)-\overline{\Omega(\xi_2,\mu_2)}}\right|,\end{align}
uniformly for $(\xi_j,\mu_j)$ in compact subsets of  $\mathcal D$ and $\|(\xi_1,\mu_1)-(\xi_2,\mu_2)\|~\geq~L^{-\tfrac12+\delta}$.

Moreover, we have
\begin{align}
\label{eq:boundR}
R(x_1,m_1,x_2,m_2)=\OO(\log L),\quad \text{as } L\to \infty,\end{align}
where the constant is uniform for $(\xi_j,\mu_j)$ in compact subsets of~$\mathcal D$. \end{proposition}
The proof of Proposition \ref{prop:asymptR1} will be postponed to Section \ref{sec:asymptoticsR}. Finally we prove Proposition \ref{prop:asymptR1}.

\begin{proof}[Proof of Proposition \ref{prop:asymptR1}]
Write $\langle h,\phi\rangle$  as $X_f$ with \[f(x,m)=-\frac{1}{L^2}\sum_{y\leq x} \triangle \phi( \Omega(x/L,m/L)) J(x/L,m/L).\] Then  $D f=-L^{-2} \triangle \phi$. Hence by $\var \langle h-\EE h,\phi\rangle=\var \langle h, \phi\rangle=\var X_f$   and  by Proposition \ref{prop:variancelinstat} we obtain
\begin{multline*}
\var \langle h-\EE h,\phi\rangle
=-\frac{\pi}{L^4}\sum_{(x_1,m_1)}\sum_{(x_2,m_2)} \triangle \phi(\Omega(x_1/L,m_1/L))\triangle \phi (\Omega(x_2/L,m_2/L) )\\
\times R(x_1,m_1,x_2,m_2) J(x_1/L,m_1/L) J(x_2/L,m_2/L).
\end{multline*}
The right-hand side can be viewed as a Riemann sum and hence by the fact that $\phi$ has compact support we have that 
\begin{multline*}
\lim_{L\to \infty}  \var\langle h-\EE h,\phi\rangle=-\iiiint \triangle \phi(\Omega(\xi_1,\mu_1)) \triangle 
\phi(\Omega(\xi_2,\mu_2))\\
\times \frac{1}{2\pi}\log \left|\frac{\Omega(\xi_1,\mu_1)-{\Omega(\xi_2,\mu_2)}}{\Omega(\xi_1,\mu_1)-\overline{\Omega(\xi_2,\mu_2)}}\right| J(\xi_1,\mu_1)J(\xi_2,\mu_2)  {\rm d}\xi_1{\rm d}\mu_1{\rm d}\xi_2{\rm d}\mu_2.
\end{multline*}
Here we used \eqref{eq:limR1} for points $(\xi_j,\mu_j)$ that are far apart and \eqref{eq:boundR} to show that the contribution of the points that are close is negligible.
By a change of variables we obtain 
\begin{multline*}\lim_{L\to \infty}  \var\langle h-\EE h,\phi\rangle\\=-\iint_{\HH \times \HH} \triangle \phi(\Omega_1) \triangle 
\phi(\Omega_2) \frac{1}{2\pi}\log \left|\frac{\Omega_1-{\Omega_2}}{\Omega_1-\overline{\Omega}_2}\right| {\rm d}m(\Omega_1) {\rm d}m(\Omega_2),\end{multline*}
where ${\rm d}m$ stands for the planar Lebesgue measure.
The fact of the matter is that  the logarithm in the integrand is the Green's function for the Laplace operator on the upper half plane with Dirichlet boundary conditions. Therefore
\[\lim_{L\to \infty}  \var\langle h-\EE h,\phi\rangle =-\iint_{\HH} 
\phi(\Omega)\triangle \phi(\Omega) {\rm d}m(\Omega)=\|\phi\|_\nabla^2,\]
where the last step follows by  integration by parts. 
\end{proof}

\section{Gaussian fluctuations}\label{sec:proof}

In this section we present a proof of Theorem \ref{th:main}.   The starting point is again that we can write $\langle h,\phi\rangle$ as a linear statistic $X_f$ with \[f(x,m)=-\frac{1}{L^2}\sum_{y\leq x} \triangle \phi( \Omega(x/L,m/L)) J(x/L,m/L),\]
as given in Lemma \ref{lem:linearstat}. Note that the support of $f$ is unbounded (in contrast to the compact support of  $\phi$). Indeed, the function $f$ is constant on horizontal rays that are at the right of the support of $\phi$ (scaled with $L$).  

The unbounded support of $f$ turns out to be inconvenient in the proof. Therefore, we split the function $f=f_1+f_2$ such that $f_2$ has bounded support containing the  (scaled) support of $\triangle \phi \circ \Omega$.  To this end, fix $\eps>0$ and define
\begin{align}\label{def:deps} \mathcal D_\eps =\{(\xi,\mu) \in \mathcal D : \Im \Omega(\xi,\mu)>\eps\}.\end{align}
Now split  $f$ as 
\begin{align}\label{eq:deffj}
f=f_1+f_2, \qquad f_1=f \chi_{\mathcal D_\eps}.
\end{align}
where $\chi_{\mathcal D_\eps}$ stands for the characteristic function of the set $\mathcal D_\eps$.  The point  is to choose $\eps$ small enough such that $\mathcal D_\eps$ contains the  support of  $\triangle \phi \circ \Omega$. See also Figure~\ref{fig:deps}. In that case, the function $f_2$ contains all the information of $\phi$.

In Section \ref{sec:varxf1} we first prove that the contribution of $X_{f_1}$ to $X_{f}$  is negligible. More precisely, the variance tends to zero  as $\eps \downarrow 0$ after taking the limit $L\to \infty$. Then in Section \ref{sec:varxf2} we prove that the fluctuation of $X_{f_2}$ are Gaussian. Finally, based on these two results and a standard probability argument we prove Theorem \ref{th:main} in Section \ref{sec:proofmainresult}  

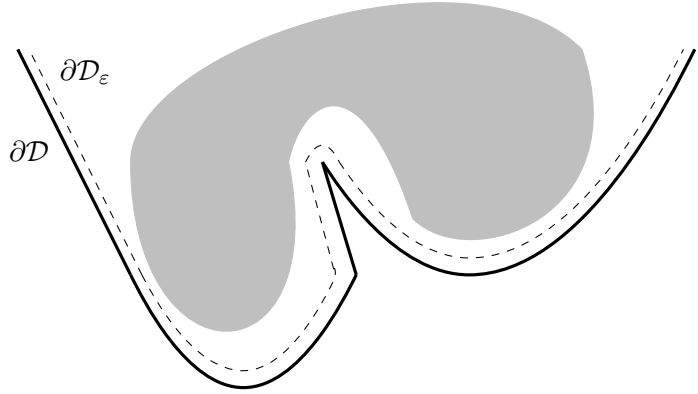
\begin{figure}[t]
\centering{
\begin{tikzpicture}[scale=1.5]
\draw[very thick] (-1,1) parabola bend (0,0) (1,1);
\draw[very thick]  (0.7,2) parabola bend (2,1) (4,3);
\draw[very thick]  (1,1) -- (0.7,2);
\draw[very thick]  (-1,1) .. controls (-1.5,2) and (-1.5,2) .. (-2,3);
\draw[dashed]  (-0.9,1) parabola bend (0,0.15) (0.8,0.95);
\draw[dashed] (0.7,2) parabola bend (2,1) (4,3);
\draw[dashed] (0.8,1.05) .. controls (0.82,1.1) and (0.82,0.9) .. (0.8,.95);
\draw[dashed] (0.8,1.05) -- (0.55,2);
\draw[dashed] (-0.9,1) .. controls (-1.4,2) and (-1.4,2) .. (-1.9,3);
\draw[dashed] (0.55,2) parabola bend (0.7,2.15) (0.85,2);
\draw[dashed] (0.85,2) parabola bend (2,1.15) (3.9,3);
\draw (-1.9,2.1) node {$\partial \mathcal D$};
\draw(-1.4,2.8) node {$\partial \mathcal D_\eps$};
\filldraw[lightgray] (-1,2) .. controls (-1,0) and (0.8,0) .. (0.4,2) -- (0.4,2) .. controls (0.5,2.5) and (1,3) .. (1.5,1.5)--
(1.5,1.5) .. controls (2,1) and (3.5,1.5) .. (3,3) --  (3,3) .. controls (2,4) and (-1,3) .. (-1,2);
\end{tikzpicture}
}
\caption{The solid curve represents $\partial \mathcal D$ and the dashed curve  $\partial \mathcal D_\eps$. The shaded region is the (compact) support of $\triangle \phi \circ \Omega$.}
\label{fig:deps}. \end{figure}

\subsection{The variance of $X_{f_1}$}  \label{sec:varxf1}
For $X_{f_1}$ we prove the following lemma, for which the main ingredient is Proposition \ref{prop:asymptR1}. 
\begin{lemma}\label{lem:varianceXf1} For $\eps>0$, let $X_{f_1}$ and ${f_1}$ be as in \eqref{eq:linstat} and \eqref{eq:deffj}. There exists a positive function $g_\phi$, with 
$\lim_{\eps \downarrow 0} g_\phi(\eps)=0$, such that
\[
\limsup _{L\to \infty} \var X_{f_1}\leq \big(g_\phi(\eps)\big)^2.
\]
 \end{lemma}
\begin{proof}
By Proposition \ref{prop:variancelinstat} we need to compute $D f_1$. We recall that $\phi$ has compact support, so that we can choose $\eps>0$ small enough so that $\triangle \phi\circ \Omega$ has a support that is entirely contained in $\mathcal D_\eps$.  In that case we have that $Df_1(x,m)\neq 0$ only if $(x,m)$ is in  the set \begin{multline}B=\{(x,m)\in L\mathcal D_\eps \mid (x-1,m) \in L \mathcal \mathcal D\setminus D_\eps\}\\
\cup 
\{(x,m)\in L\left(\mathcal D\setminus \mathcal D_\eps\right) \mid (x-1,m) \in L \mathcal D_\eps\}.\end{multline}
Note that $B$ is a discrete set of points close to $L\partial \mathcal D_\eps$. 
Moreover,
\begin{align}\label{eq:dfx1b}
Df_1(x,m) = \pm\frac{\sqrt \pi}{L^2} \sum_{y\leq x\pm1}\triangle \phi(\Omega(y/L,m/L)) J(x/L,m/L) 
\end{align}
if $(x,m)\in B$ and $Df_1(x,m)=0$ otherwise. Hence \begin{align}\label{varianceB}
\var X_{f_1}=\sum_{(x_1,m_1)\in B}\sum_{(x_2,m_2)\in B} Df_1(x_1,m_1) D f_1(x_2,m_2) R(x_1,m_1,x_2,m_2).
\end{align} 
The right-hand side is a Riemann sum for a double integral over $\partial \mathcal D_\eps\times \partial \mathcal D_\eps$.

First note that the right-hand side of \eqref{eq:limR1},  for points $(\xi_j,\mu_j)\in \partial \mathcal D_\eps$, can be written as\begin{multline}
\label{alternativeGreen}
-\frac{1}{2\pi^2} \log\left|\frac{\Omega(\xi_1,\mu_1)-\Omega(\xi_2,\mu_2) }{\Omega(\xi_1,\mu_1)-\overline{\Omega(\xi_2,\mu_2)}}\right| 
=\frac{1}{2\pi^2} \log\left|1+\frac{2{\rm i} \eps}{\Omega(\xi_1,\mu_1)-\Omega(\xi_2,\mu_2)}\right|\\=\frac{1}{4\pi^2}  \log\left|1+\frac{4 \eps^2}{|\Omega(\xi_1,\mu_1)-\Omega(\xi_2,\mu_2)|^2}\right| \end{multline}
where in the last inequality we  used that $\Omega(\xi_1,\mu_1)-\Omega(\xi_2,\mu_2)\in \R$ for $(\xi_j,\mu_j)\in \partial \mathcal D_\eps$. 

In order to interpret \eqref{varianceB}  as a Riemann sum for an integral over $\partial \mathcal D_\eps\times \partial \mathcal D_\eps$, we use the following parametrization for $\mathcal D_\eps$  that is based on the projection onto the vertical coordinate.  There exists a continuous function $(\eps,t) \mapsto (\xi_\eps(t),\mu_\eps(t))$ and points  $-\infty=t_0<t_1^{(\eps)}<t_2^{(\eps)}<t_3^{(\eps)}<t_4=\infty$ that depend continuously on $\eps$, such that  $\mu'_\eps(t)$ is constant on $(t_j^{(\eps)},t_{j+1}^{(\eps)})$ and takes the values $\pm 1$, and
\[
\mathcal D_{\eps}=\{ (\xi_\eps(t),\mu_\eps(t)) \mid t\in \R\}.
\]
By \eqref{eq:dfx1b} we see that the limiting behavior of  $\pm LDf_1([L\xi_\eps(t)],[L \mu_\eps(t)])$ is given by
\begin{align}\label{delta} \delta(t)=\sqrt{\pi}\int_{\xi_\eps(t)}^\infty \triangle \phi(\Omega(\xi,\mu_\eps(t)) J(\xi,\mu_\eps(t)) {\rm d}\xi.\end{align}
By taking the limit $L\to \infty$ in \eqref{varianceB} and using \eqref{alternativeGreen} and \eqref{delta} we obtain \begin{multline*}
\limsup_{L\to \infty}  \var X_{f_1}\\
\leq \frac{1}{4\pi}\int_\R \int_\R |\delta  (t)| | \delta (s)| 
  \log\left|1+\frac{4 \eps^2}{|\Omega(t)-\Omega(s)|^2}\right|  {\rm d} s{\rm d}t=:\big(g_\phi(\eps)\big)^2.
  \end{multline*}
Note that since $\phi$ has compact support, also $\delta$ has compact support and hence the $g_\phi(\eps)<\infty$. By taking the limit, we have  $\lim_{\eps \downarrow 0} g_\phi(\eps)=0$ and this proves the statement.
\end{proof}

\subsection{Gaussian fluctuations for $X_{f_2}$} \label{sec:varxf2}
The purpose of this paragraph is to prove the following proposition.
\begin{proposition} \label{lem:gfluctf2}
Let $f_2$ be as in \eqref{eq:deffj}. Then 
\begin{align}\label{eq:lemgfluct}
\lim_{L \to \infty} \left(
\EE \left({\rm e}^{{\rm i} t(X_{f_2}-\EE X_{f_2})}\right) -{\rm e}^{-\frac{1}{2} t^2 \var X_{f_2}}\right)=0,
\end{align}
uniformly for $t$ in compact subsets of $\C$.
\end{proposition}

The proof of Proposition \ref{lem:gfluctf2} is based on some asymptotic results on the kernel $K$, that we will first present. We recall that for $p\geq 1$ the $p$-th Schatten norm of a matrix $A$ is defined as 
\begin{align}
\|A\|_p^p=\Tr |A^*A|^{p/2}=\sum \sigma_j(A)^p,
\end{align}
where  $\sigma_j(A)$ are the singular values of $A$, counted according to multiplicity. Moreover, $\|A\|_\infty=\sigma_{max}(A)$, i.e. the maximal singular value.
\begin{lemma} \label{lem:conj} Let $\delta>0$ and $\mathcal D_c$ be a compact subset of $\mathcal D$.   There exists a function $G:\left( \Z\times \N\right) \cap L\mathcal D_c \to (0,\infty)$ such that the matrix 
\begin{align}
\label{eq:defG}
K_G(x_1,m_1,x_2,m_2)=\frac{G(x_1,m_1)}{G(x_2,m_2)} K(x_1,m_1,x_2,m_2)\end{align}
for $(x_1,m_1),(x_2,m_2) \in L \mathcal D_c$ satisfies
\begin{enumerate}
\item $\|K_G\|_2=\OO(L^{1+\delta})$ as $L\to \infty$.
\item $\|K_G\|_4=\OO(L)$as $L\to \infty$.
\item $\|K_G\|_\infty=\OO(L)$ as $ L\to \infty$. 
\end{enumerate}
Moreover, if we set $\Omega_j=\Omega(x_j/L,m_j/L)$ and $F_j(\cdot)=F(\cdot \mid x_j/L,m_j/L)$, then
\begin{enumerate}
\item[4.] for  $(x_j,m_j)\in L\mathcal D_c$  with $\|(x_1,m_1)-(x_2,m_2)\|> L^{\tfrac{1}{2}+\delta}$ we have
\begin{multline}\label{eq:asymptoticsI24terms}
K(x_1,m_1,x_2,m_2)=-\frac{1}{2\pi}\Big(\frac{{\rm e}^{L(F_1(\Omega_1) -F_2(\Omega_2))} }{\Omega_1(\Omega_1-\Omega_2){L} (-F''_1(\Omega_1))^{1/2} (F''_2(\Omega_2))^{1/2}}\\
+\frac{{\rm e}^{L(F_1(\Omega_1) -F_2(\overline\Omega_2))} }{\Omega_1(\Omega_1- \overline\Omega_2){L} (-F''_1(\Omega_1))^{1/2} (F''_2(\overline\Omega_2))^{1/2}}\\
+\frac{{\rm e}^{L(F_1(\overline\Omega_1) -F_2(\Omega_2))} }{\overline\Omega_1(\overline\Omega_1-\Omega_2){L} (-F''_1(\overline\Omega_1))^{1/2} (F''_2(\Omega_2))^{1/2}}\\
+\frac{{\rm e}^{L(F_1(\overline\Omega_1) -F_2(\overline\Omega_2))} }{\overline\Omega_1(\overline\Omega_1-\overline\Omega_2){L} (-F''_1(\overline\Omega_1))^{1/2} (F''_2(\overline\Omega_2))^{1/2}}\Big)(1+\OO(L^{-\delta/2})), 
\end{multline}
as $L\to \infty$, where the constant is uniform for $(x_j,m_j)\in L\mathcal D_c$.  The signs of the square roots $(-F''_1(\Omega_1))^{1/2}$ and $(F''_2(\Omega))^{1/2}$ are given in Lemma \ref{lem:I21}. \end{enumerate}
\end{lemma}

The proof of this lemma relies on a steepest descent analysis on the integral representation \eqref{eq:defKdoubleint} for the kernel $K$. This analysis requires a significant amount of work and we therefore postpone it to Section \ref{sec:asymptoticsK}.  The choice in the sign of the square roots  depends on a deformation of the contour in the steepest descent analysis. Since the the sign does not matter for the proofs in this section we omit  it here and postpone it to Lemma \ref{lem:I21}. 

In the proof of Proposition \ref{lem:gfluctf2} we will also need the following.
\begin{lemma} \label{lem:boundonf} Let $f$ be as in Lemma \eqref{lem:linearstat} and $f=f_1+f_2$ as in \eqref{eq:deffj}. Then 
\begin{enumerate}
\item $\|f_2\|_\infty =\OO(L^{-1})$ as $L \to \infty$.
\item $\|f_2\|_2=\OO(1)$ as $L\to \infty$. 
\end{enumerate} 
\end{lemma}
\begin{proof}
This follows easily by the definitions of $f_2$ in \eqref{eq:deffj} and $f$ in \eqref{eq:deff},  and the fact that $\|\triangle \phi\|_\infty<\infty$ by assumption.
\end{proof}
The last ingredient for the proof is the following result that is based on a lemma in \cite{Kenyon}.
\begin{lemma}We have that
\begin{align}\label{eq:kenyon}
\lim_{L\to \infty} \Tr (f_2K)^\ell=0, \qquad \ell \geq 3.
\end{align}
\end{lemma}
\begin{proof}
The proof of this statement follows by \eqref{eq:asymptoticsI24terms} and  using the same arguments as in  \cite[Sec. 7]{Kenyon}. We will  briefly highlight the main points here. To start with, expand the trace 
\begin{multline}\label{eq:expandedtracekenyon}
\Tr (f_2 K)^\ell =\sum_{(x_1,m_1)}\cdots \sum_{(x_\ell,m_\ell)} f(x_1,m_1)\cdots f(x_\ell,m_\ell)\\
\times K(x_1,m_1,x_2,m_2)K(x_2,m_2,x_3,m_3) \cdots K(x_\ell,m_\ell,x_1,m_1).\end{multline} 
Note that the support of $f_2$ is contained in $L\mathcal D_c$ for some compact subset of $\mathcal D$.  Moreover,  the contribution of points $(x_j,m_j)$ that are close is negligible. To this end, we note that $\|f_2\|_\infty=\OO(L^{-1})$ and that we can replace $K$ with $K_G$ and use \eqref{splitKG}, \eqref{eq:IG1a}, \eqref{eq:IG21} and \eqref{eq:IG22} to show that the contribution   coming from points for which $\|(x_j,m_j)-(x_{j+1},m_{j+1})\|\leq L^{1/2+\delta}$ for some $1\leq j \leq \ell$,  tends to zero as $L\to \infty$. 

Hence in \eqref{eq:expandedtracekenyon} we only need to consider points $(x_j/L,m_j/L)\in \mathcal D_c$ for which the subsequent point are sufficiently far apart such that we can use \eqref{eq:asymptoticsI24terms}. By expanding parenthesis we obtain $4^\ell$ terms. Most terms are highly oscillating and therefore their contribution to the sum \eqref{eq:expandedtracekenyon} is negligible. Hence there are only $4\ell$ surviving terms, that do not oscillate.  Up to multiplication with a symmetric function, each term can be written as 
\begin{align} \label{kenyonfinal} \prod_{j=1}^{\ell} \frac{1}{\omega_j-\omega_{j+1}},\end{align}
where $\omega_j=\Omega_j$ or $\overline{\Omega}_{j}$ and $\omega_{\ell+1}=\omega_1$.
The point  of the proof is that the sum  \eqref{eq:expandedtracekenyon} is invariant under permutation of  variables. Hence we can replace $\omega_j$ in \eqref{kenyonfinal} with $\omega_{\sigma(j)}$ for any permutation $\sigma$. Moreover, we can replace it by a sum over any set of permutations.   Lemma 7.3 in \cite{Kenyon} tells us that the sum of \eqref{kenyonfinal} over  $\ell$-cycles is zero and hence we obtain the statement. \end{proof}

We are now ready for the proof of Proposition  \ref{lem:gfluctf2}. The proof relies on a Fredholm determinant identity for the characteristic function of a linear statistic. The statement then follows by estimates on the operator in the determinant. 

\begin{proof}[Proof of Proposition \ref{lem:gfluctf2}]
First we note that since $X_{f_2}$ is a  linear statistic, it is standard  that we can write the characteristic function as a (Fredholm) determinant 
\begin{align}\label{eq:chartodet}
\EE \left(\exp {\rm i} t X_{f_2} \right)=\det \left(1+({\rm e}^{{\rm i} t f_2}-1)K\right).
\end{align}
Indeed, by writing the exponential of the sum as the product of the exponentials we obtain
\[\EE \left(\exp {\rm i} t X_{f_2} \right)=\EE  \left( \prod_{(x,m)\in \mathcal C } \left(1+({\rm e}^{{\rm i} t f_2(x,m)}-1) \right)\right)
\]
where $\mathcal C$ at the right-hand side is a random configuration of points. By expanding the product we see that the right-hand side is a Fredholm determinant (see for example \cite{J} for more details).  

 Note that  $f_2$ has support in $L \mathcal D_c$ some compact $\mathcal D_c\subset \mathcal D$. Hence the same is true for ${\rm e}^{{\rm i} tf_2}-1$. This implies that we can (and do) restrict the matrix $K$ to the points on the grid $\Z\times \N $ that are inside $L\mathcal D_c$. In particular we can apply the results of Lemma  \ref{lem:conj}.
 
 Since the determinant is invariant under conjugation with any invertible matrix we have that 
\begin{align}\label{eq:chartodetconj}
\det \left(1+({\rm e}^{{\rm i} t f_2}-1)K\right)=\det \left(1+({\rm e}^{{\rm i} t f_2}-1)K_G,\right)
\end{align}
where $K_G$ is as in \eqref{eq:defG}. It is useful to replace $K_G$ to $K$ since we can control  norms of $K_G$ as $L\to \infty$ as given in Lemma \ref{lem:conj}.  In the remaining part of the proof, we analyze the asymptotic behavior of the determinant at the right-hand side of \eqref{eq:chartodetconj}.  

Let us first make some remarks on the operator in the determinant. First, by Lemma \ref{lem:boundonf} there exists a constant $C_1$ such that
\begin{align}
\|{\rm e}^{i t f_2}-1\|_\infty \leq t C_1  L^{-1}.
\end{align}
By combining this with the norms on $K_G $ in Lemma \ref{lem:conj} and using the fact that $\|AB\|_p\leq \|A\|_\infty\|B\|_p$, we see that there exists constant $C_2$ and $C_3$ such that 
\begin{align}\label{eq:inequalitiesnorm}
\|({\rm e}^{i t f_2}-1)K_G\|_\infty \leq t C_2 \quad \text{and}\quad  \|({\rm e}^{i t f_2}-1)K_G\|_4\leq t C_3.\end{align}
Clearly, the constants $C_j$ do not depend on $L$.

Now we return to \eqref{eq:chartodetconj}. The determinant can be rewritten as
\begin{align}\label{eq:expanddet}
\det\left(1+ \left({\rm e}^{i t f_2}-1\right)K_G\right)=\exp\left(
\sum_{\ell=1}^\infty \frac{(-1)^{\ell+1}}{\ell}
\Tr \left( ({\rm e}^{i t f_2}-1)K_G\right)^\ell\right).\end{align}
This expansion is valid for $t$ in a  neighborhood of the origin in the complex plane. A priori, this neighborhood depends on  $L$. Our first task is to show that the neighborhood can be chosen independent of $L$. To this end we note that for $\ell\geq 4$ we have 
\begin{multline*}
\left|\Tr \left( ({\rm e}^{i t f_2}-1)K_G\right)^\ell\right|\leq \|({\rm e}^{i t f_2}-1)K_G\|_\ell^\ell\\ 
\leq  \|({\rm e}^{i t f_2}-1)K_G\|_\infty^{\ell-4} \|({\rm e}^{i t f_2}-1)K_G\|_4^4
\end{multline*}
where in the last inequality we iteratively used the inequality $\|A\|_{p+1}^{p+1}\leq \|A\|_p^p \|A\|_\infty$ valid for $A$ in the $p+1$ Schatten class. Now by \eqref{eq:inequalitiesnorm}  we then have that there exists constants $C_4,D>0$, independent of $L$, such that 
\begin{align}
\label{eq:boundontrace}
\left|\Tr \left( ({\rm e}^{i t f_2}-1)K_G\right)^\ell\right|\leq t^\ell C_4^\ell D.\end{align}
This means that \eqref{eq:expanddet} is valid for $|t|< C_4^{-1}$, for all $L$.

The next step is to analyze the asymptotic behavior of the traces in the exponential at the right-hand side of \eqref{eq:expanddet}. We claim that 
\begin{align}\label{eq:tracelimit}
\lim_{L\to \infty}  \Tr \left( ({\rm e}^{i t f_2}-1)K_G\right)^\ell=0, \qquad \ell\geq 3.\end{align}
To this end,  we first observe that we have the following inequality 
\begin{align}\label{eq:tracelimitineq}
|\Tr f_2^{s_1} K_{G} \cdots f_2^{s_\ell} K_G | \leq \|f_2\|_\infty^{\sum s_j}  \|K_G\|_2^2 \|K_G\|_\infty^{\ell-2}, \quad \ell\geq 2.
\end{align}
Here we used the well-known identities $|\Tr A|\leq \|A\|_1$, $\|AB\|_1\leq \|A\|_2\|B\|_2$ and $\|AB\|_2 \leq \|A\|_2\|B\|_\infty$ where $\|\cdot\|_1$ stands for the trace norm. If $s_1+s_2+\ldots s_\ell \geq \ell+1$ the trace converges to zero as $L\to \infty$ by Lemma \ref{lem:conj}(i)+(iii) and Lemma \ref{lem:boundonf}(i).  Hence, by expanding the exponential, we see that  the only term in  \eqref{eq:tracelimit}   that  possibly remains in the limit, is  therefore the term with $s_j=1$. Hence we proved
\begin{align}
\label{eq:tracelimit2}\lim_{L\to \infty} \left( \Tr \left( ({\rm e}^{i t f_2}-1)K_G\right)^\ell-\Tr (i t f_2K_G)^\ell \right)=0, \qquad \ell\geq 3.
\end{align}
Now \eqref{eq:tracelimit} follows by \eqref{eq:tracelimit2} and \eqref{eq:kenyon}.

By  \eqref{eq:expanddet} and  \eqref{eq:tracelimit} we have
\begin{align}\label{eq:limitexpanddet}
\lim_{L \to \infty} \det\left(1+ ({\rm e}^{i t f_2}-1)K_G\right) {\rm e}^{-\Tr  ({\rm e}^{i t f_2}-1)K_G+\frac{1}{2} \Tr  \left(({\rm e}^{i t f_2}-1)K_G\right)^2} =1,
\end{align}
uniformly for $t$ in a neighborhood of the origin. Note that  we also changed the order of the limits, which is allowed by \eqref{eq:boundontrace}.

 We proceed by simplifying \eqref{eq:limitexpanddet} a bit further. 

By the identities after \eqref{eq:tracelimitineq} and  Lemma \ref{lem:boundonf} we have 
\[|\Tr f^\ell K_G|\leq   \|f\|_\infty^{\ell-1} \|f\|_2 \|K_G\|_2=\OO(L^{2+\delta -\ell}), \quad \text{as } L\to \infty, \quad \ell\geq 3,\]
and hence by expanding the exponential we obtain 
\begin{align} \label{eq:tracelimitsimpl1}
\Tr ({\rm e}^{i t f_2}-1)K_G ={\rm i} t \Tr f_2 K_G -\frac{t^2}{2} \Tr f_2^2K_G+\OO(L^{\delta-1})
\end{align}
as  $L\to \infty$.  Moreover, by \eqref{eq:tracelimitineq} we have 
\begin{align} \label{eq:tracelimitsimpl2}
\Tr \left(({\rm e}^{i t f_2}-1)K_G\right)^2=-{t^2}\Tr (f_2K_G)^2+\OO(L^{2\delta-1}).
\end{align}
Inserting \eqref{eq:tracelimitsimpl1} and \eqref{eq:tracelimitsimpl2} in \eqref{eq:limitexpanddet} leads to
\begin{align}
\lim_{L \to \infty} \det\left(1+ \left({\rm e}^{i t f_2}-1\right)K_G\right) {\rm e}^{-{\rm it }\Tr f_2 K_G+\frac{t^2}{2} \Tr (f_2^2K_G-f_2K_G f_2K_G)} =1.
\end{align}
By \eqref{eq:meanoflinstat} and \eqref{eq:varianceoflinstat} and the fact that we can replace $K_G$ back to $K$ in the trace we can write this as
\begin{align}
\lim_{L \to \infty} \det\left(1+ \left({\rm e}^{i t f_2}-1\right)K_G\right) {\rm e}^{- {\rm it} \EE X_{f_2}+\frac{t^2}{2} \var X_{f_2}} =1,
\end{align}
Since the variance of $X_{f_2}$ remains bounded as $L\to \infty$ (indeed, $X_{f_2}=X_{f}-X_{f_1}$ and both $\var X_f$ and $\var X_{f_1}$ are bounded), we also have 
\begin{align}\label{eq:limitlocal}
\lim_{L \to \infty} \left(\det\left(1+ \left({\rm e}^{i t f_2}-1\right)K_G\right) {\rm e}^{- {\rm it} \EE X_{f_2}} -{\rm e}^{-\frac{t^2}{2} \var X_{f_2}} \right)=0,
\end{align}
uniformly for $t$ in a neighborhood of the origin. Combining \eqref{eq:limitlocal} with \eqref{eq:chartodet} gives \eqref{eq:lemgfluct} in a neighborhood of the origin.

To prove that \eqref{eq:limitlocal} (and hence \eqref{eq:lemgfluct}) holds for all $t\in \C$ we argue as follows. Since \eqref{eq:tracelimitsimpl1} and \eqref{eq:tracelimitsimpl2} hold for all $t\in \C$ it is sufficient to prove that \eqref{eq:limitexpanddet} holds for all $t\in \C$. To this end, we note that
\begin{multline}\label{eq:deffL}
f_L(t):= \det\left(1+ ({\rm e}^{i t f_2}-1)K_G\right) {\rm e}^-{\Tr  ({\rm e}^{i t f_2}-1)K_G+\frac{1}{2} \Tr  \left(({\rm e}^{i t f_2}-1)K_G\right)^2} \\
 = {\det}_4\left(1+ ({\rm e}^{i t f_2}-1)K_G\right) {\rm e}^{\frac{1}{3}\Tr  \left(({\rm e}^{i t f_2}-1)K_G\right)^3}
\end{multline}
where ${\det}_4$ is the regularized determinant of order 4 (in fact, this is how the regularized determinant is defined \cite{Simon}). For such determinants we have the inequality \cite[Th. 9.2(b)]{Simon}
\[|{\det}_4(1+A)| \leq \exp \Gamma_4 \|A\|_4,\]
for some constant $\Gamma_4$ (independent of $A$).
Applying this inequality with $A=({\rm e}^{i tf_2} -1)K_G$ in \eqref{eq:deffL} and using \eqref{eq:inequalitiesnorm} and \eqref{eq:tracelimit},  we see that $\{f_L\}_L$  is a sequence of entire functions that is locally bounded. By Montel's  Theorem,  it has a subsequence converging uniformly on compact subsets of $\C$. By analyticity and \eqref{eq:limitexpanddet}, the limit of any convergent subsequence must be the constant function $1$. Hence, the entire sequence converges to $1$ and  \eqref{eq:limitlocal} holds uniformly for  $t$ in compact subsets of $\C$. Combining this with \eqref{eq:chartodet} gives the statement.
\end{proof}

\subsection{Proof of Theorem \ref{th:main}}\label{sec:proofmainresult}
\begin{proof}[Proof of Theorem \ref{th:main}]
In view of Lemma \ref{lem:linearstat} we can reformulate the statement as
\begin{align}\label{eq:mainstat2}
\lim_{L\to \infty} \EE \exp({\rm i}t (X_f-\EE X_f))=\exp\left(-\tfrac{t^2}{2} \|\phi\|_\nabla^2\right), \quad t\in \R,
\end{align}
with $f$ as in Lemma \ref{lem:linearstat}.   Split $f=f_1+f_2$ as in \eqref{eq:deffj}, define $Y_{f_j}=X_{f_j}-\EE X_{f_j}$ and write 
\begin{multline}\label{eq:finaltriangle}
\left|\EE \exp({\rm i} t Y_f)-\exp(-\tfrac{1}{2} t^2 \|\phi\|_\nabla ^2) \right|
\leq
\left|\EE \exp({\rm i} t Y_{f})-\EE \exp({\rm i} t Y_{f_2}) \right|\\
+\left|\EE \exp({\rm i} t Y_{f_2})-\exp(-\tfrac{1}{2}t^2 \var Y_{f_2}) \right|\\
+\left|\EE \exp(-\tfrac{1}{2} t^2 \var Y_{f_2})-\exp(-\tfrac{1}{2}t^2 \|\phi\|_\nabla ^2) \right|.
\end{multline}
We estimate the three terms at the right-had side separately.

Since $Y_f$ and $Y_{f_2}$ are real valued and $\EE Y_f=\EE Y_{f_2}=0$, we have
\begin{align}
\begin{split}
\left|\EE \exp({\rm i} t Y_f)-\EE \exp({\rm i} t Y_{f_2}) \right|\leq \EE\left|\exp({\rm  i}t Y_f) -\exp({\rm i} t Y_{f_2})\right|
\\
\leq |t| \EE(|Y_f-Y_{f_2}|) 
\leq |t| \sqrt{\var (Y_f-Y_{f_2})}=|t|\sqrt{\var Y_{f_1}}.\end{split}
\end{align}
Hence by Lemma \ref{lem:varianceXf1}  there exists a positive function $g_\phi(\eps)$ with $\lim_{\eps \downarrow 0} g_\phi(\eps)=$ and
\begin{align}\label{eq:limsup1}
\begin{split}\limsup_{L\to \infty} 
\left|\EE \exp({\rm i} t Y_{f})-\EE\exp({\rm i} t Y_{f_2}) \right| \leq  |t| g_\phi(\eps).
\end{split}
\end{align}
This estimates the first term at the right-hand side of \eqref{eq:finaltriangle}.  

Proposition \ref{lem:gfluctf2} deals with the second term.

As for the  third term, note that 
\begin{align}
\var Y_{f} =\var( Y_{f_1}+Y_{f_2})=\var Y_{f_1}+2\EE(Y_{f_1}Y_{f_2})+\var Y_{f_2}.
\end{align}
Hence by applying Cauchy-Schwarz and Lemma \ref{lem:varianceXf1} we see that there exists a constant $A>0$ such that \begin{multline}\label{eq:limsup2}
\limsup_{L\to \infty} \left|\var Y_{f_2}-\var Y_f\right|\\
\leq  \limsup_{L\to \infty}\left( 2 \sqrt{\var Y_{f_1} \var Y_f} +\var Y_{f_1}\right)\leq A g_\phi(\eps).\end{multline}
Finally, inserting \eqref{eq:limsup1}, \eqref{eq:limsup2},  \eqref{eq:propvariance} and \eqref{eq:lemgfluct}  in \eqref{eq:finaltriangle} leads to 
\begin{align}
\limsup_{L\to \infty} \left|\EE \exp({\rm i} t Y_f)-\exp(-\tfrac{1}{2} t^2 \|\phi\|_\nabla ^2) \right|
\leq (A+|t|)g_\phi(\eps),
\end{align}  Since this holds for arbitrary $\eps>0$ we can take the limit $\eps \downarrow 0$ and we arrive at the statement.
\end{proof}

\section{Asymptotic analysis of $K$} \label{sec:asymptoticsK}

In this section we prove Lemma \ref{lem:conj} and Theorem \ref{th:macro}. The proofs are based on classical steepest descent arguments on the integral representation for the kernel \eqref{eq:defKdoubleint}, which
by \eqref{eq:defF} we can write as
\begin{align}\label{eq:defKdoubleint2}
\begin{split}
K(x_1,m_1,x_2,m_2) =-\frac{1}{2\pi {\rm i}} \oint_{\Gamma_0} {\rm e}^{L (F_1(w)-F_2(w))}\frac{{\rm d}w}{w}\\
+\frac{1}{(2 \pi{\rm i})^2} \oint_{\Gamma_0} \oint_{\Gamma_{1,2}}{\rm e}^{L(F_1(w)-F_2(z))}\frac{ {\rm d}z {\rm d}w}{w(w-z)}.
\end{split}
\end{align}
Here we used lower index $F_j$ to denote the function $F$ in \eqref{eq:defF} with parameters $(\xi_j/\mu_j)=(x_j/L,m_j/L)$. We will also write $\Omega_j$ to denote $\Omega(x_j/L,m_j/L)$.

 We will only need the asymptotics for $(x,m)\in L \mathcal D_c$ for a compact subset $\mathcal D_c\subset \mathcal D$ in the proof of Lemma \ref{lem:conj}.  For more background on steepest descent method applied to double integrals, we refer to \cite{O}.

\subsection{Saddle points of $F$}

We start by investigating the saddle points of $F$,  as well as the paths of steepest descent and ascent for $\Re F$ leaving from these points. 

 If $\mu\leq \mu_0$ then $F$ has either two real saddle points or two complex conjugates and no more.  If $\mu>\mu_0$  there is an additional real saddle point in the open interval $(1,2)$.  We recall that $\mathcal D$ is defined as the set of all pairs $(\xi,\mu)$ such that we have complex conjugate solutions. Here we are interested in asymptotics for $K$ for points $(x_j,m_j)\in L\mathcal D_c$ and hence we restrict our attention to   the situation where there are two complex conjugate saddle points   (and  one real saddle point in the interval $(1,2)$ in case $\mu>\mu_0$). 
 
 Since $F$ is analytic in the upper half plane, the paths of steepest descent for $\Re F$ are (parts of)  level lines for $\Im F$. The non-real saddle points are simple, resulting in a quadratic behavior for $F$ near these points. This implies that we have  two  paths of steepest descent leaving from each non-real saddle point. One of the paths ends up in either $w=1$ or $w=2$  and the other ends in $\infty$ and is asymptotically parallel with the negative part of the real axis. There are also two paths of steepest ascent leaving from the saddle points. One of them is ending in $w=0$ and the other in $\infty$ asymptotically parallel to the positive part of the real axis.  
 
 Finally, if $\mu>\mu_0$ there is a third saddle point in the interval $(1,2)$. The paths of steepest descent are just the parts of the interval $(1,2)$ to the left and right of that point. The paths of steepest ascent end either in $w=0$ or $\infty$. If the paths of steepest descent of the non-real saddle points end in $w=1$ then the path of steepest ascent leaving from the saddle point in $(1,2)$ ends in $\infty$. Otherwise it ends in $w=0$.  See also Figure~\ref{fig:realsaddlepoints}.

 \begin{figure}[t]
\centering
\includegraphics[scale=0.3]{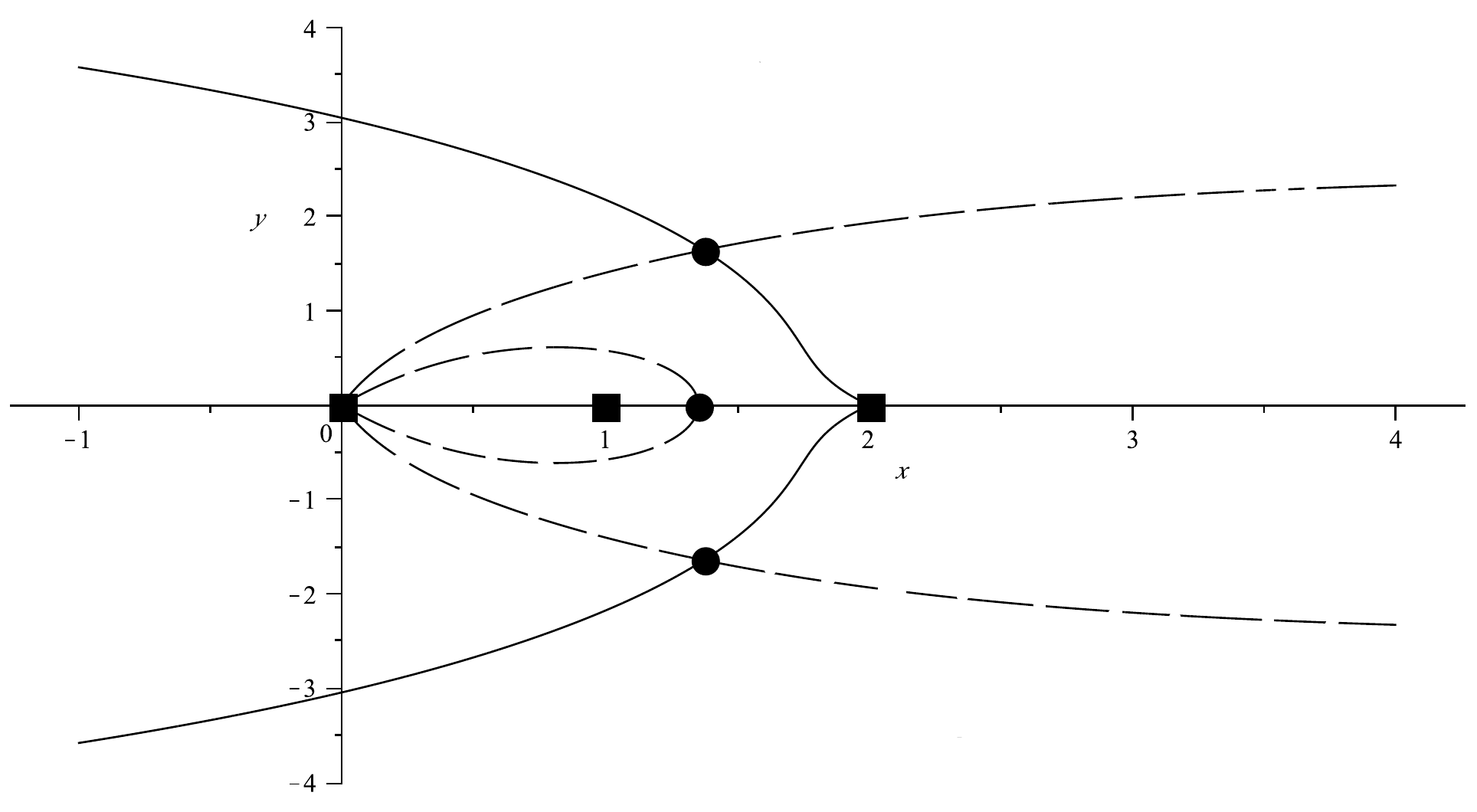}
\includegraphics[scale=0.3]{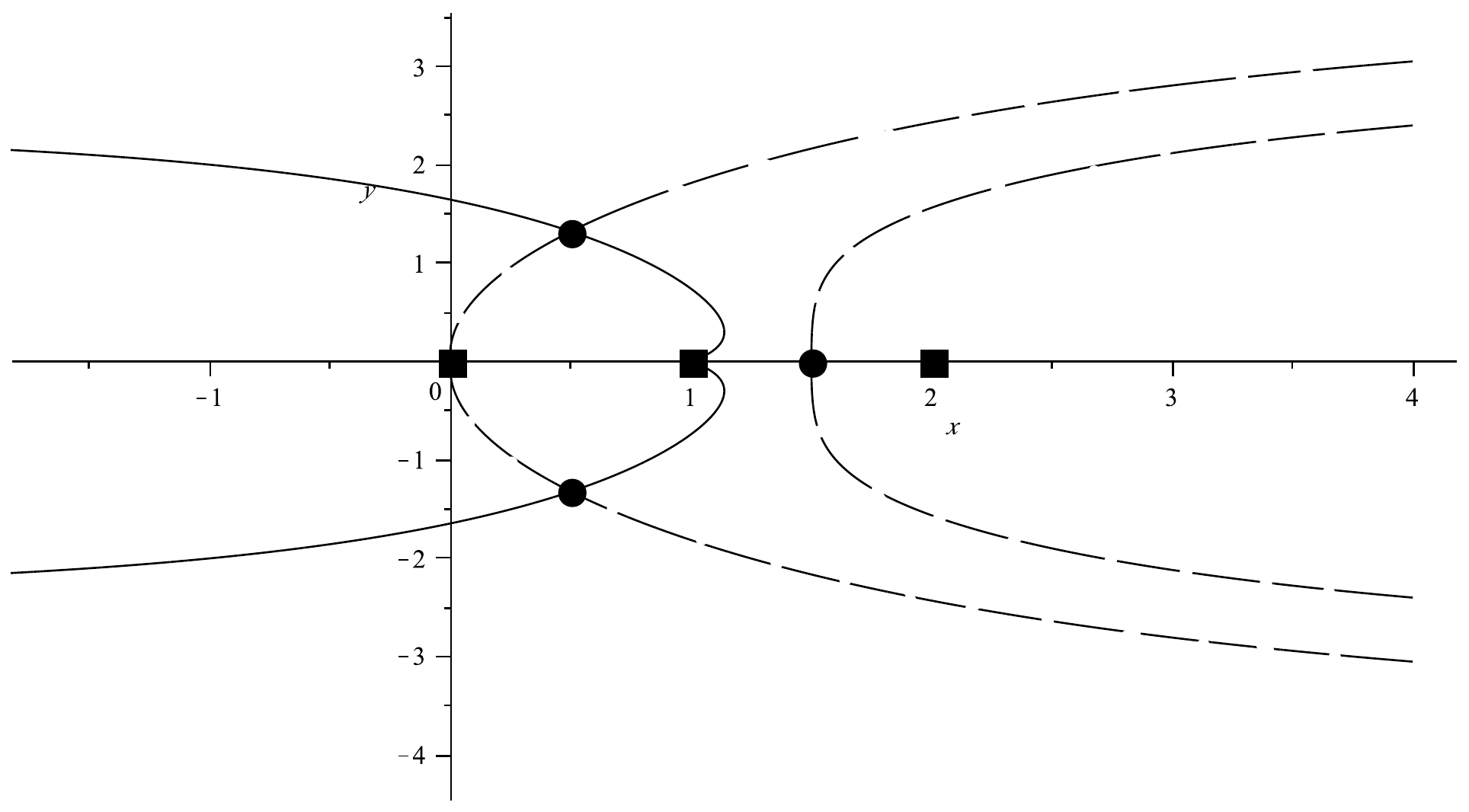}
\caption{Contours of steepest descent/ascent for $\Re F$ leaving from the saddle points  for the situation $\mu>\mu_0$. The black squares represent the points $w=0$, $w=1$ and $w=2$. The black dots are the saddle points.}
\label{fig:realsaddlepoints}

\end{figure}

\subsection{Deforming the contours}
The next step in the steepest descent analysis is to  deform $\Gamma_0$ and $\Gamma_{1,2}$ to be paths of steepest descent/ascent for $\Re F$ leaving from  saddle points $\Omega_1$ and $\Omega_2$. However, by deforming the contours in this way we have to take the term $(w-z)^{-1}$ into account.   Indeed, the deformed contours $\Gamma_0$ and $\Gamma_{1,2}$ will intersect at a point, say  $\eta$, in the upper half plane (and hence by symmetry also at $\overline{\eta}$ in the lower half plane). See also Figure~\ref{fig:deforming}. This implies that after  deforming the contours we pick up a residue and we obtain 
\begin{align}
\begin{split}
K(x_1,m_1,x_2,m_2) =\frac{1}{2\pi {\rm i}} \int_{C_\eta} {\rm e}^{L (F_1(w)-F_2(w))}\frac{{\rm d}w}{w}\\+\frac{1}{(2 \pi{\rm i})^2} \oint_{\Gamma_0} \oint_{\Gamma_{1,2}}{\rm e}^{L(F_1(w)-F_2(z))}\frac{ {\rm d}z {\rm d}w}{w(w-z)}.
\end{split}
\end{align}
Here $C_\eta$ is a path  connecting $\overline{\eta}$ with $\eta$  that crosses the real axis on the positive side if $m_1\geq m_2$ and on the negative side if $m_1<m_2$.  In both cases, the orientation of the path is from $\overline{\eta}$ to $\eta$.

Let us  introduce the notations
\begin{align}\label{eq:defI1}
I_1=\frac{1}{2\pi {\rm i}} \int_{C_\eta} {\rm e}^{L (F_1(w)-F_2(w))}\frac{{\rm d}w}{w},
\end{align}
and
\begin{align}\label{eq:defI2}
I_2=\frac{1}{(2 \pi{\rm i})} \oint_{\Gamma_0} \oint_{\Gamma_{1,2}}{\rm e}^{L(F_1(w)-F_2(z))}\frac{ {\rm d}z {\rm d}w}{w(w-z)},
\end{align}
and  deal with the terms $I_1$ and $I_2$ separately.

\subsection{Analysis of the double integral $I_2$} First we consider the integral $I_2$. In the first lemma that we now present  we assume that $\Omega_1$ and $\Omega_2$ are not close to each other. 

\begin{lemma} \label{lem:I21} Let  $\delta>0$ and  $\mathcal D_c\subset\mathcal D$ compact. If $(x_j,m_j)\in L\mathcal D_c$  and $\|(x_1,m_1)-(x_2,m_2)\|> L^{\tfrac{1}{2}+\delta}$, then  
\begin{multline}
I_2(x_1,m_1,x_2,m_2)=-\frac{1}{2\pi}\Big(\frac{{\rm e}^{L(F_1(\Omega_1) -F_2(\Omega_2))} }{\Omega_1(\Omega_1-\Omega_2){L} (-F''_1(\Omega_1))^{1/2} (F''_2(\Omega_2))^{1/2}}\\
+\frac{{\rm e}^{L(F_1(\Omega_1) -F_2(\overline\Omega_2))} }{\Omega_1(\Omega_1- \overline\Omega_2){L} (-F''_1(\Omega_1))^{1/2} (F''_2(\overline\Omega_2))^{1/2}}\\
+\frac{{\rm e}^{L(F_1(\overline\Omega_1) -F_2(\Omega_2))} }{\overline\Omega_1(\overline\Omega_1-\Omega_2){L} (-F''_1(\overline\Omega_1))^{1/2} (F''_2(\Omega_2))^{1/2}}\\
+\frac{{\rm e}^{L(F_1(\overline\Omega_1) -F_2(\overline\Omega_2))} }{\overline\Omega_1(\overline\Omega_1-\overline\Omega_2){L} (-F''_1(\overline\Omega_1))^{1/2} (F''_2(\overline\Omega_2))^{1/2}}\Big)(1+\OO(L^{-\delta/2})), 
\end{multline}
as $L\to \infty$. The constant in  $\mathcal O$ is uniform for $(\xi_j,\mu_j)\in \mathcal D_c$ and does not depend on $\delta$.
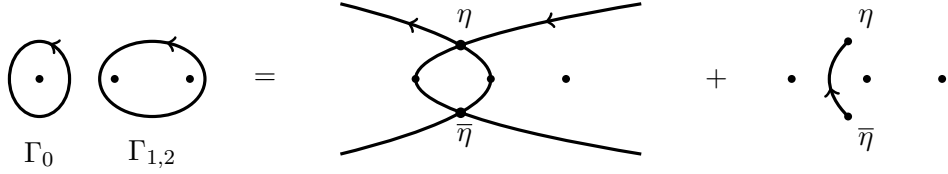
\begin{figure}[t]
\begin{tikzpicture}[scale=1]
]
\draw[postaction={decorate},decoration={
  markings,
  mark=at position 0.2 with {\arrow{>}}},very thick] (-1,0) ellipse (0.4 and 0.5); 
\draw[postaction={decorate},decoration={
  markings,
  mark=at position 0.2 with {\arrow{>}}},very thick] (0.5,0) ellipse (0.7 and 0.5);
\fill (-1,0) circle (1.5pt);
\fill (0,0) circle (1.5pt);
\fill (1,0) circle (1.5pt);

\fill (4,0) circle (1.5pt);
\fill (5,0) circle (1.5pt);
\fill (6,0) circle (1.5 pt);
\draw[postaction={decorate},decoration={markings,
  mark=at position 0.2 with {\arrow{>}}},very thick,rotate around ={-90:(4,0)} ] (3,3) parabola   bend (4,0) (5,3);
\draw[postaction={decorate},decoration={markings,
  mark=at position 0.8 with {\arrow{>}}},very thick,rotate around ={90:(5,0)} ] (4,2) parabola bend (5,0) (6,2);
\draw[postaction={decorate},decoration={markings,
  mark=at position 0.4 with {\arrow{>}}},very thick,rotate around ={90:(9.5,0)} ] (9,-.25) parabola bend (9.5,0) (10,-.25);
  \draw (2,0) node {$=$};
   \draw (8,0) node {$+$};
\fill (9,0) circle (1.5pt);
\fill (10,0) circle (1.5pt);
\fill (11,0) circle (1.5pt);
\fill[rotate around ={90:(10,0)}] (9.5,.25) circle (1.5pt);
\fill[rotate around ={90:(10,0)}] (10.5,.25) circle (1.5pt);
\draw[rotate around ={90:(10,0)}] (10.8,.25) node[right] {$\eta$};
\draw[rotate around ={90:(10,0)}] (9.2,.25) node[right] {$\overline{\eta}$};

\draw[rotate around ={90:(5,0)}] (5.8,.35) node {$\eta$};
\fill[rotate around ={90:(5,0)}] (5.45,.4) circle (2pt);
\fill[rotate around ={90:(5,0)}] (4.55,.4) circle (2pt);
\draw[rotate around ={90:(5,0)}] (4.55,.35) node[below] {$\overline{\eta}$};
\draw (-1,-1) node {$\Gamma_0$};
\draw (0.5,-1) node {$\Gamma_{1,2}$};
\end{tikzpicture}
\caption{When deforming the contours $\Gamma_0$ and $\Gamma_{1,2}$ in the double integral in \eqref{eq:defKdoubleint2} we pick up a residue that leads to an additional single integral.} 
\label{fig:deforming}
\end{figure}
The square root $(-F_1''(\Omega_1))^{1/2}$ is chosen such that the orientation of the line $w=\Omega_1+\frac{s}{(-F_1''(\Omega_1))^{1/2})}$ tangent to $\Gamma_0$ at $\Omega_2$, when traversed from $-\infty$ to $\infty$, coincides with the orientation of $\Gamma_{0}$ at $\Omega_1$.  Similarly for $(-F_1''(\overline{\Omega}_1))^{1/2}$.

The square root $(F_2''(\Omega_2))^{1/2}$ is chosen such that the orientation of the line $w=\Omega_1+\frac{t}{(F''_2(\Omega_2))^{1/2})}$ tangent to $\Gamma_{1,2}$, when traversed from $-\infty$ to $\infty$, coincides with the orientation of $\Gamma_{1,2}$ at $\Omega_2$. Similarly for $(
F_2''(\overline{\Omega}_2))^{1/2}$.
\end{lemma}
\begin{proof}
The point of the steepest descent method is that the main contribution of the integral comes from small neighborhoods around the saddle points. The parts of the contours $\Gamma_{0}$ and $\Gamma_{1,2}$ that are sufficiently far from the saddle points only give a contribution that is exponentially small.  To obtain the leading order terms, we  introduce local variables around the saddle points. Since we have a double integral and each integral gives rise two (conjugate) saddle points, there are four of such leading terms. 

Let us consider one of these four terms, namely the parts of the integrals around the saddle points $\Omega_1$ and $\Omega_2$.  In  small neighborhoods around these saddle points we slightly deform the contours $\Gamma_0$ and $\Gamma_{1,2}$ to their tangent lines and introduce the local variables 
\begin{align}\label{eq:localvar}
w=\Omega_1 + \frac{t}{\sqrt{L} (-F_1''(\Omega_1))^{1/2}} \textrm{   and   } z=\Omega_1 + \frac{s}{\sqrt{L} (F_2''(\Omega_2))^{1/2}},\end{align}
for $t\in [-L^{\delta/2},L^{\delta/2}]$ and 
$s\in [-L^{\delta/2},L^{\delta/2}]$ where $\delta>0$ is the fixed constant in the Lemma. Note that with these new variables we have
\begin{align}\label{eq:Taylor}
\lim_{L\to \infty} L(F_1(w(t))-F_1(\Omega_1))=-\tfrac{t^2}{2} \textrm{   and   }\lim_{L\to \infty} (F_2(z(s))-F_2(\Omega_2))=\tfrac{s^2}{2}, 
\end{align}
uniformly for $t\in [-L^{\delta/2},L^{\delta/2}]$ and by continuity also uniformly for $(\xi_j,\mu_j)\in \mathcal D_c$. Note that at the endpoints $t=\pm L^\delta$ we have that ${\rm e}^{L (F_1(w(t))-F_1(\Omega_1))}$ is exponentially small as $L\to \infty$. Moreover, from these points the contours are continued to paths of steep descent/ascent. This implies that the parts of the double integral that are not in the neighborhood of $\Omega_1$ and $\Omega_2$ (and their conjugates) is exponentially small as $L\to \infty$. It therefore remains to compute the asymptotic behavior of the integral
\begin{align}\label{eq:steepestdouble}
\begin{split}
\frac{1}{
(2\pi i)^2 L (-F''_1(\Omega_1))^{1/2}  (F''_2(\Omega_2))^{1/2}  }\int _{-L^\delta}^{L^\delta}  \int_{-L^\delta}^{L^\delta} {\rm e}^{L(F_1(w(t))-F_2(z(s)))} \frac{{\rm d}s {\rm d}t}{w(t)(w(t)-z(s))}
\end{split}\end{align}
The main issue that remains is to handle the term $1/(w(t)-z(s))$. We claim that since   $\Omega_1$ and $\Omega_2$ are sufficiently far apart by assumption,  we have
\begin{align}\label{eq:wminzfar}
\frac{1}{w(t)-z(s)}=\frac{1}{\Omega_1-\Omega_2}(1+\OO(L^{-\delta/2})), \quad L\to \infty.
\end{align}
in the new  variables. Indeed, from \eqref{eq:inverse1} and \eqref{eq:inverse2} it is not hard to check that the inverse map $\Omega \mapsto (\xi,\mu)$ is Lipschitz and hence there is a constant $A$ such that \begin{align}\label{lipschitz}
|\Omega_1-\Omega_2|\geq A \| (\xi_1,\mu_1)-(\xi_2,\mu_2)\|>AL^{-1/2+\delta}.\end{align} Combining this with the fact that  $F''(\Omega)$ is bounded from below for $(\xi,\mu)\in \mathcal D_c$ we have
\[| \Omega_1-\Omega_2| |F''_{1,2}(\Omega_{1,2})|^{1/2} \sqrt{L} \geq \tilde A L^{\delta}\]
for some constant $\tilde A$. Hence 
\[\frac{t}{| \Omega_1-\Omega_2| |F''_{1,2}(\Omega_{1,2})|^{1/2} \sqrt{L}}\leq \tilde A^{-1} L^{-\delta/2}\]
for $t\in [-L^{\delta/2},L^{\delta/2}]$. This proves \eqref{eq:wminzfar}. Note that by continuity the constant $\tilde A^{-1}$ can be chosen independent of $(\xi_j,\mu_j)\in \mathcal D_c$ and hence depends only on $\mathcal D_c$ but not on $\delta>0$. 

Now substituting \eqref{eq:wminzfar} into \eqref{eq:steepestdouble} leads to the following
\begin{multline*}
\frac{1}{
(2\pi i)^2  L(-F_1''(\Omega_1))^{1/2} (F_2''(\Omega_2))^{1/2}  }\int _{-L^\delta}^{L^\delta}  \int_{-L^\delta}^{L^\delta} {\rm e}^{L(F_1(w(t))-F_2(z(s)))} \frac{{\rm d}s {\rm d}t}{w(t)(w(t)-z(s))}\\
=\frac{{\rm e}^{L(F_1(\Omega_1)-F_2(\Omega_2))} }{
(2\pi i)^2  L(-F_1''(\Omega_1))^{1/2} (F_2''(\Omega_2))^{1/2} (\Omega_1-\Omega_2) \Omega_1} \\ \times \int _{-L^\delta}^{L^\delta} \int_{-L^\delta}^{L^\delta} {\rm e}^{L\left(F_1(w(t))-F_1(\Omega_1))-(F_2(z(s))-F_2(\Omega_2))\right)} {\rm d}s{\rm  d}t (1+\OO(L^{-\delta/2})).
\end{multline*}
By taking $L$ large and computing the Gaussian integrals gives
\[
\begin{split}
\frac{1}{
(2\pi i)^2  L(-F_1''(\Omega_1))^{1/2} (F_2''(\Omega_2))^{1/2}  }\int _{-L^\delta}^{L^\delta}  \int_{-L^\delta}^{L^\delta} {\rm e}^{L(F_1(w(t))-F_2(z(s)))} \frac{{\rm d}s {\rm d}t}{w(t)(w(t)-z(s))}\\
=-\frac{{\rm e}^{L(F_1(\Omega_1)-F_2(\Omega_2))} }{
2\pi   L(-F_1''(\Omega_1))^{1/2} (F_2''(\Omega_2))^{1/2} (\Omega_1-\Omega_2) \Omega_1}  (1+\OO(L^{-\delta/2})).
\end{split}\]
The three other parts coming from the conjugate points $\Omega_1$ and $\Omega_2$ can be dealt with in the same way. This four terms together give the leading asymptotic term as given in the statement.
 \end{proof}

If $\Omega_1$ and $\Omega_2$ are sufficiently close,  we need the  inequality that is formulated in the following lemma. 
\begin{lemma}\label{lem:I22}
Let $\delta>0$ and $\mathcal D_c\subset \mathcal D$ compact. Then there exists a constant such that
\begin{align}
|I_2(x_1,m_1,x_2,m_2)|\leq \frac{C {\rm e}^{L\Re(F_1(\Omega_1)-F_2(\Omega_2))}}{\sqrt{L}}.
\end{align}
 for $(x_j,m_j)\in L\mathcal D_c$ with $\|(x_1,m_1)-(x_2,m_2)\|\leq L^{\tfrac{1}{2}+\delta}$. \end{lemma}
\begin{proof}
The proof follows by the same approach as in Lemma \ref{lem:I21}. However, in this case \eqref{eq:wminzfar} is no longer valid and we need to estimate this term in a different way. Note that since $(\xi_j,\mu_j)$ are close, we also have that $F_1''(\Omega_1)=F''_2(\Omega_2)(1+\OO(L^{-1/2+\delta}))$, where the constant is independent of $(\xi_j,\mu_j)\in \mathcal D_c$ as long as they are close. 
But then $(-F_1(\Omega_1))^{1/2}=\pm {\rm i} (F_2(\Omega_2))^{1/2}$ where the $\pm$-sign depends on the choice of the square roots.  Hence we have 
\begin{align}
\frac{1}{w(t)-z(s)}= \frac{1}{\sqrt{L} (-F_1(\Omega_1))^{1/2}} \frac{1}{\Omega_1-\Omega_2+t\pm {\rm  i}s +\OO(L^{-1/2+\delta})}.\end{align}
This means the double integral has a possible singularity (note if $(\xi_1,\mu_1)=(\xi_2,\mu_2)$ it certainly does), however this singularity is integrable.  By proceeding as in Lemma \ref{lem:I21},  computing the Gaussian integrals and using  the fact that $ F''(\Omega(\xi,\mu))$ is bounded from below for $(\xi,\mu)\in\mathcal D_c$, we arrive at the statement.
\end{proof}
\subsection{Analysis of the single integral $I_1$}

Next we give a bound for the single integral.  

\begin{lemma} \label{lem:I1} Let $I_1$ as in \eqref{eq:defI1} and $\mathcal D_c\subset \mathcal D$ compact.  Then there exists a constant $C>0$ such that 
\begin{align} |I_1(x_1,m_1,x_2,m_2)|\leq \frac{C{\rm e}^{L \Re\left(F_1(\eta)-F_2(\eta)\right)}}{1+R},
\end{align}
for $(x_j,m_j)\in L\mathcal D_c$, where $R=\sqrt{(x_1-x_2)^2+(m_1-m_2)^2}$.
\end{lemma}
\begin{proof}
Let us first consider the case $m_1\geq m_2$.  In that case, deform the contour to an arc from $\overline{\eta}$ and $\eta$ that is a part of the circle centered around the origin with radius $|\eta|$.  A short calculation using \eqref{eq:defF} and the fact that $m_1\geq m_2$ shows that 
\begin{align} \label{eq:I1est0}
\frac{{\rm d}}{{\rm d}t} \Re \left(F_1(\eta {\rm e}^{-{\rm i}t})-F_2(\eta {\rm e}^{-{\rm i}t})\right)<0,\qquad 0\leq t<\arg \eta.
\end{align}
It also follows by the geometric principle that the distance of the point  $\eta {\rm e}^{-{\rm i}t}$ to $w=1$ and $w=2$ decreases when $t$ increases from $0$ to $\arg \eta$ and attains its minimum at $t=\arg \eta$. It implies that 
\begin{align} \label{eq:I1est1}
\frac{1}{2\pi}\left|\int_{C_\eta} {\rm e}^{L(F_1(w)-F_2(w))} \frac{{\rm d}w}{w} \right|\leq {\rm e}^{L \Re (F_1(\eta)-F_2(\eta))}.
\end{align}
We will refine the inequality for points $(x_j,m_j)$ that are far apart using standard  steepest descent arguments.  Set 
\[(x_1-x_2,m_1-m_2)= R (\cos \phi,\sin \phi)\]
and assume $R>0$. Then write
\begin{align}\label{eq:relGandF} L(F_1(w)-F_2(w))=R G_{1,2}(w,\phi).\end{align}
so that 
\begin{align} \label{eq:ftog} \frac{1}{2\pi}\int_{C_\eta} {\rm e}^{L(F_1(w)-F_2(w))}  \frac{{\rm d}w}{w} =\frac{1}{2\pi} 
\int_{C_\eta} {\rm e}^{R G_{1,2}(w,\phi)} \frac{dw}{w}.
\end{align}

Note that $\eta$ depends continuously on $(\xi_j,\mu_j)\in \mathcal D_c$. By compactness of $\mathcal D_c$, there exists an $r>0$ be sufficiently small so that $\Im (\eta {\rm e}^{-{\rm i}r})>0$ for all $(x_j,m_j)\in L \mathcal D_c$. By \eqref{eq:I1est0}   there is a constant $A>0$ such that 
\begin{align}\label{eq:estimateG}
\Re\left( G_{1,2}(\eta {\rm e}^{-{\rm i}t},\phi)-G_{1,2}(\eta)\right)\leq-At, \qquad 0\leq t\leq r.
\end{align}
Again by continuity and compactness, the constant $A$ can be chosen independent of $\phi$ and $\eta$,  depending on $\mathcal D_c$ only.

Finally, we split the arc $C_\eta$ in three pieces, two small parts near $\eta$ and $\overline{\eta}$ and a third remaining part. A simple estimate  on \eqref{eq:ftog} gives
\begin{multline}\label{eq:I1est2} \frac{1}{2\pi} \left| 
\int_{C_\eta} 
{\rm e}^{R\Re G_{1,2}(w,\phi)} \frac{dw}{w}\right|\\
\leq \frac{1}{2\pi} \int_0^r \left({\rm e}^{R \Re G_{1,2}(\eta {\rm e}^{-{\rm i} t} ,\phi)}+{\rm e}^{R G_{1,2}(\overline{\eta} {\rm e}^{{\rm i} t},\phi )} \right){\rm d}t+
{\rm e}^{-Ar R} {\rm e}^{L\Re G_{1,2}(\eta,\phi)}. 
\end{multline}
For the second term in the second term at the righthand side we used  \eqref{eq:I1est0},  \eqref{eq:relGandF}  and \eqref{eq:estimateG}.
To analyze the two parts near $\eta$ and $\overline{\eta}$ we use \eqref{eq:estimateG} and obtain
\begin{align}\label{eq:I1est3} \int_0^r {\rm e}^{R \Re G(\eta {\rm e}^{-{\rm i} t} )}{\rm d}t\leq {\rm e}^{R \Re G(\eta)} \int_0^r {\rm e}^{-R At} {\rm d}t\leq  \frac{{\rm e}^{R \Re G_{1,2}(\eta,\phi)}}{AR}.  \end{align}
Substituting  \eqref{eq:I1est3} in the right-hand side  of   \eqref{eq:I1est2}, inserting  \eqref{eq:relGandF}   and combing the result with \eqref{eq:I1est1} gives that statement for the case $m_1\geq  m_2$.

If  $m_1 <m_2$, then  the statement follows by a similar argument. The main difference is that we deform  $C_\eta$ to the arc of  the circle with radius $|\eta|$ that lies at the left of the origin.
\end{proof}

\subsection{Proof of Lemma \ref{lem:conj}}

\begin{proof}
The function $G$ that we use to define $K_G$ as in \eqref{eq:defG} is given by 
\begin{align}
G(x,m)={\rm e}^{L\Re F(\Omega)}.
\end{align}
Let us first  give some estimates on $K_G$ that can be deduced from the estimates on $K$ that we have obtained so far. To this end we split 
\begin{align}\label{splitKG} 
K_G=I_{1,G}+I_{2,G}
\end{align} where 
\begin{align}
I_{j,G}(x_1,m_1,x_2,m_2)=\frac{G(x_2,m_2)}{G(x_1,m_1)} I_j(x_1,m_1,x_2,m_2),
\end{align}
where $I_j$ are as defined in \eqref{eq:defI1} and \eqref{eq:defI2}. 

From Lemma \ref{lem:I1} we deduce
 \begin{align}\label{eq:IG1}
|I_{1,G}|\leq \frac{A}{1+R} 
 {\rm e}^{L \Re\left(F_1(\eta)-F_1(\Omega_1)\right)-L\Re \left(F_2(\eta)-F_2(\Omega_2)\right)},
\end{align}
where $R=\|(x_1,m_1)-(x_2,m_2)\|$. Note that since $\eta$ is an intersection point of $\Gamma_0$ and $\Gamma_{1,2}$ which are paths of steepest descent/ascent we have that 
\[ \Re\left(F_1(\eta)-F_1(\Omega_1)\right)-\Re \left(F_2(\eta)-F_2(\Omega_2)\right)\leq 0.\]
Hence, from \eqref{eq:IG1} we in have in particular
 \begin{align}\label{eq:IG1a}
|I_{1,G}|\leq \frac{A}{1+R}.
\end{align}
Suppose that  $R>L^{1/2+\delta}$. Then also $\Omega_1$ and $\Omega_2$ are far apart and hence $\eta$ is far from either $\Omega_1$ or $\Omega_2$. More precisely, there exists a constant $\tilde B$ such that  $|\eta-\Omega_j|\geq \tilde B L^{-1/2+\delta}$ for $j=1$ or $j=2$ (see also \eqref{lipschitz}). From \eqref{eq:IG1} and the quadratic behavior of $F$ near the saddle point,  we obtain  that there exists a constant $B>0$ such that 
 \begin{align}\label{eq:IG1b}
|I_{1,G}|\leq \frac{A}{1+R} 
 {\rm e}^{-BL^{2 \delta}},
\end{align}
for $(x_j,m_j)\in L\mathcal D_c$.

As for $I_{2,G}$ we deduce from Lemma's \ref{lem:I21} and \ref{lem:I22} that
 \begin{align}\label{eq:IG21}
\left|I_{2,G}(x_1,m_1,x_2,m_2)\right|\leq
 \frac{A_1}{L  |\Omega_1-\Omega_2|},
 \end{align}
for $(x_j,m_j)\in L\mathcal D_c$ with $\|(x_1,m_1)-(x_2,m_2)\|> L^{1/2+\delta}$.  
Moreover,
\begin{align}\label{eq:IG22}
\left|I_{2,G}(x_1,m_1,x_2,m_2)\right|\leq
 \frac{A_2}{\sqrt{L}}, 
 \end{align}
 for $(x_j,m_j)\in L\mathcal D_c$ with $\|(x_1,m_1)-(x_2,m_2)\|\leq  L^{1/2+\delta}$.

\vskip 0.1in

Now we come to the proofs of (i), (ii), (iii) and (iv). 

\vskip 0.1in

(i). We  split the Hilbert-Schmidt norm into two parts using the triangular inequality $\|K_G\|_2\leq \|I_{1,G}\|_2+\|I_{2,G}\|_2$.

Let us start with the  term $\|I_{1,G}\|_2$. Then by \eqref{eq:IG1a} and   a short calculation  we obtain \begin{align}\label{eq:HS1}
\sum_{(x_1,m_1)\in L\mathcal D_c}  \sum_{(x_2,m_2)\in L\mathcal D_c} \left| I_{1,G}(x_1,m_1,x_2,m_2) \right|^2=\OO( L^{2}\log L),
\end{align}
as $ L\to \infty.$

Next we analyze $\|I_{2,G}\|_2$. If $(x_1,m_1)$ and $(x_2,m_2)$ are close, then we have by \eqref{eq:IG22}
\begin{align}\label{eq:HS2}
\sum_{(x_1,m_1)\in L\mathcal D_c } \sum_{\overset{(x_2,m_2) \in L \mathcal D_c}{\|(x_1,m_1)-(x_2,m_2)\|\leq L^{1/2+\delta} }} |I_{2,G}(x_1,m_1,x_2)|^2
= \OO(L^{2+\delta}),
\end{align}
as $L\to \infty$. 

Finally, we consider points $(x_1,m_1)$ and $(x_2,m_2)$ that are sufficiently far from each other. Then we use \eqref{eq:IG21} to obtain
\begin{multline}
\sum_{(x_1,m_1)\in L\mathcal D_c } \sum_{\overset{(x_2,m_2) \in L \mathcal D_c}{\|(x_1,m_1)-(x_2,m_2)\|> L^{1/2+\delta} }}\left|I_{2,G}(x_1,m_1,x_2,m_2)\right|^2\\
\leq
\sum_{(x_1,m_1)\in L\mathcal D_c } \sum_{\overset{(x_2,m_2) \in L \mathcal D_c}{\|(x_1,m_1)-(x_2,m_2)\|> L^{1/2+\delta} }} \frac{1}{L^2 |\Omega_1-\Omega_2|^2}
\end{multline}
We interpret the right-hand side as a Riemann sum  and  estimate it by the integral 
\begin{multline}
\sum_{(x_1,m_1)\in L\mathcal D_c } \sum_{\overset{(x_2,m_2) \in L \mathcal D_c}{\|(x_1,m_1)-(x_2,m_2)\|> L^{1/2+\delta} }}\left|I_{2,G}(x_1,m_1,x_2,m_2)\right|^2\\
\leq L^{2} A \iiiint_{\Omega(\mathcal D_c)\times (\Omega(\mathcal D_c) \setminus B_{\Omega_1}) }\frac{1}{|\Omega_1-\Omega_2|^2}{\rm d}m(\Omega_1) {\rm d} m(\Omega_2),
\end{multline}
where $B_{\Omega_1}$ is a ball around $\Omega_1$ with a radius that is of order $L^{1/2+\delta}$,  $A>0$ is some constant and ${\rm d}m$ stands for the two dimensional Lebesgue measure. Note that we also used that, by continuity and the fact that the Jacobian does not vanish,  the Jacobian in  \eqref{eq:Jacobian}  is bounded from below and above on compact subsets. 
As $L\to \infty$ the integral at the right-hand side grows logarithmically, so in particular we have
\begin{align}\label{eq:HS3}
\sum_{(x_1,m_1)\in L\mathcal D_c } \sum_{(x_2,m_2) \in L( \mathcal D_c\setminus B_{\Omega_1})} \left|I_{2,G}(x_1,m_1,x_2,m_2)\right|^2
=\OO(L^{2}\log L),
\end{align}
as $L\to \infty$.  The statement now follows by combining \eqref{eq:HS1}, \eqref{eq:HS2} and \eqref{eq:HS3}.

\vskip 0.1in

(ii) The proof is  similar to the proof of (i).  Write
\begin{multline}
\|K_G\|_4^4 =\Tr K_G^* K_G K_G^*K_G\\ \leq 
\sum \cdots \sum \big|K_G(x_2,m_2,x_1,m_1)
K_G(x_2,m_2,x_3,m_3) \\ \times
K_G(x_4,m_4,x_3,m_3)
K_G(x_4,m_4,x_1,m_1)\big|,
 \end{multline}
 and insert \eqref{eq:IG1a}--\eqref{eq:IG22}. A short computation shows that the contribution of points that are close are negligible (in contrast to the double sum in the Hilbert-Schmidt norm) and the leading term comes from points $(x_j,m_j)$ that are far apart. For such points we use \eqref{eq:IG1b} to deduce that the  term $I_{1,G}$ only leads to an exponentially small contribution. Hence by \eqref{eq:IG21} and arguing as in (i) we see that there exists a constant $A>0$ such that
 \[\|K_G\|_4^4 \leq A L^4 \iiiint \frac{{\rm d} m(\Omega_1) {\rm d} m(\Omega_2) {\rm d} m(\Omega_3) {\rm d} m(\Omega_4)}{|\Omega_1-\Omega_2| |\Omega_2-\Omega_3| |\Omega_3-\Omega_4| |\Omega_4-\Omega_1|}. \]
Despite the singularities in the integrand, the latter integral is bounded and we arrive at the statement. Note the difference with (i) where the integral, due to the singularity,   grows logarithmically with $L$. 
 
 \vskip 0.1 in
 
 (iii) This follows from (ii) and the fact that $\|\cdot\|_p\geq \|\cdot\|_\infty$ for $p\geq 1$. 
 
\vskip 0.1 in

(iv) In  Lemma \ref{lem:I21} we proved the statement for $I_{2}$ instead of $K$. However, by  adding the term $I_1$ into the right-hand side of \eqref{eq:asymptoticsI24terms} and observing that $I_{1,G}$ is exponentially small for points far apart,  we see that the statement is also true for $K$. 
\end{proof}

 \subsection{Proof of Theorem \ref{th:macro}}
 \begin{proof}[Proof of Theorem \ref{th:macro}]
 The proof is fairly standard for determinantal process with a kernel that is represented by a double contour integral. We will therefore allow ourselves to be brief.
 
 Write $K=I_1+I_2$ with $I_j$ as in \eqref{eq:defI1} and \eqref{eq:defI2}. On the diagonal the main contribution comes from $I_1$ instead of $I_2$. Indeed, if $(\xi,\mu)\in \mathcal D$, then by \eqref{eq:IG22} and the fact the $I_{2,G}=I_2$ on the diagonal, we see that $I_2(\xi,\mu,\xi,\mu)$ tends to $0$ as $L\to \infty$. On the other hand, from \eqref{eq:defI1} and the fact that  $\eta=\Omega(\xi,\mu)$ we have
 \begin{align}\label{eq:finallyyy}I_1(\xi,\mu,\xi,\mu)=\frac{1}{2\pi{\rm i}} \int_{\overline{\Omega}}^{\Omega} \frac{{\rm d}w}{w}=\frac{\arg\Omega(\xi,\mu)}{\pi}.\end{align}
 Hence we proved, with $(x,m)$ as in \eqref{eq:thmacroscale}, that
 \begin{align}\label{eq:limdensK}
 \lim_{L\to \infty} K(x,m,x,m)=\frac{\arg \Omega(\xi,\mu) }{\pi},
 \end{align}
for  $(x,m)$ such that $(\xi,\mu)\in \mathcal D$.

From the proofs as presented in this paper, we only have \eqref{eq:finallyyy} uniformly for $(\xi,\mu)$ in compact subsets of $\mathcal D$. However, by a steepest descent analysis for the points that are near the boundary $\partial \mathcal D$ or outside $\mathcal D$ it  can be shown that \eqref{eq:finallyyy} holds on $\overline{\mathcal D}$ and, if $(\xi,\mu)$ is point in the upper half plane to the right of $\mathcal D$, then $K(x,m,x,m)\to 0$ uniformly as $L\to \infty$.

Since \[\frac{1}{L} \EE h(y,m)=\frac{1}{L}\sum_{ x\geq y} K(x,m,x,m),\]
we obtain from \eqref{eq:limdensK}
\[\lim_{L\to \infty} \frac{1}{L} \EE h(y,m) =\int_\xi^\infty \arg (\Omega(\xi',\mu)){\rm d}\xi'\]
By using \eqref{eq:partialxi} this integral is easily computed and we obtain the statement.
   \end{proof}

 \section{Asymptotic analysis of $R$} \label{sec:asymptoticsR}

In this section we prove  Proposition \ref{prop:asymptR1}.  First we show that the kernel $R$ can be expressed as a quadruple integral. The asymptotics for $R$ then follows from steepest descent arguments on this quadruple integral representation.
\begin{figure}[t]
\centering{
\begin{tikzpicture}[scale=1.1]
\fill (0,0) circle (1pt);
\fill (1,0) circle(1pt);
\fill (2,0) circle (1pt);
\draw[postaction={decorate},decoration={
  markings,
  mark=at position 0.2 with {\arrow{>}}},very thick,dashed] (1,0) ellipse (1.5 and 0.65);
  \draw[postaction={decorate},decoration={
  markings,
  mark=at position 0.25 with {\arrow{>}}},very thick,dashed] (1,0) ellipse (1.7 and 0.8);
   \draw[postaction={decorate},decoration={
  markings,
  mark=at position 0.25 with {\arrow{>}}},very thick] (1,0) ellipse (2 and 1.1);

\draw[postaction={decorate},decoration={
  markings,
  mark=at position 0.25 with {\arrow{>}}},very thick] (0,0) ellipse (0.2 and 0.2);
\draw (0.5,0) node {$\Gamma_0$};
\draw (-0.9,0.8) node {$\Gamma_0'$};
\end{tikzpicture}\hspace{0.5cm}
%
}
\caption{Configuration of contours in the quadruple integral representation \eqref{eq:quadR} for $R$ in case $m_1<m_2$.}
\label{fig:quadR}
\end{figure}
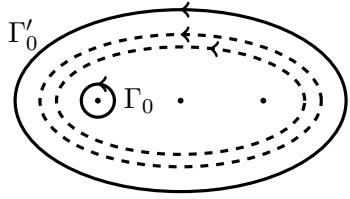

\begin{proposition}\label{prop:quadR}
With $R$  as defined in \eqref{eq:defR} and $p_m$ as defined in \eqref{eq:defpm}, we have that 
\begin{multline}\label{eq:quadR}
R(y_1,m_1,y_2,m_2)=-\frac{1}{(2\pi{\rm i})^4} \oint_{\Gamma_0} \oint_{\Gamma_{1,2}}  \oint_{\Gamma_0'}  \oint_{\Gamma_{1,2}' }\frac{{\rm e}^{t(w+w')}}{{\rm e}^{t(z+z')}} 
\\
 \times\frac{p_{m_1}(w)p_{m_2}(w')}{p_{m_2}(z)p_{m_1}(z')}
 \frac{z^{y_2+m_2}z'^{y_1+m_1}}{w'^{y_2+m_2} w^{y_1+m_1}}\frac{{\rm d}z' {\rm d}w' {\rm d}z {\rm d}w}{(w-z)(w-z')(w'-z)(w'-z')}.
\end{multline}

If $m_1<m_2$, then the contour $\Gamma_0$ goes around $w=0$, the contours $\Gamma_{1,2}$ and $\Gamma_{1,2}'$ go around $z,z'=1,2$ and $\Gamma_0$,  and finally,  the contour $\Gamma_0'$ goes around $w'=0$ and $\Gamma_{1,2}'$ and $\Gamma_{1,2}$. 

If $m_1>m_2$, then the contour $\Gamma_0'$ goes around $w'=0$, the contours $\Gamma_{1,2}$ and $\Gamma_{1,2}'$ go around $z,z'=1,2$ and $\Gamma_0'$,  and finally, the contour $\Gamma_0$ goes around $w=0$ and $\Gamma_{1,2}'$ and $\Gamma_{1,2}$.

If $m_1=m_2$, then the contours  $\Gamma_{1,2}$ and $\Gamma_{1,2}'$ go around $z,z'=1,2$, the contours  $\Gamma_0$ and $\Gamma_0'$ go  around the origin and  $\Gamma_{1,2}$ and $\Gamma_{1,2}'$ respectively. In addition, if $y_1<y_2$, then $\Gamma_{1,2}'$ goes around $\Gamma_0$ and if, on the other hand, $y_1\geq y_2$, then $\Gamma_{1,2}$ goes around $\Gamma_0'$.

All contours have counterclockwise orientation. 
\end{proposition}
\begin{proof} 
First we recall that, by deforming the contours $\Gamma_0$ or $\Gamma_{1,2}$, we can write the kernel $K$ as one double integral (and no single integral). Indeed, by arguing as in Lemma \ref{lem:KsquareK}, we deform $\Gamma_{1,2}$ such that it also encircles $\Gamma_{0}$ in case $m_1<m_2$ and, in case $m_1\geq m_2$, we deform $\Gamma_{0}$ such that it encircles $\Gamma_{1,2}$. See also Figure~\ref{fig:m1m2}. 
This means that we can write the product of the kernels as a quadruple integral
 \begin{multline} \label{eq:productofkernels}
K(x_1,m_1,x_2,m_2)K(x_2,m_2,x_1,x_2)=\frac{1}{(2\pi{\rm i})^4} \oint_{\Gamma_0} \oint_{\Gamma_{1,2}}  \oint_{\Gamma_0'}  \oint_{\Gamma_{1,2}'} \\
 \frac{{\rm e}^{t(w+w')}}{{\rm e}^{t(z+z')}} 
\frac{p_{m_1}(w)p_{m_2}(w')}{p_{m_2}(z)p_{m_1}(z')}\frac{z^{x_2+m_2}z'^{x_1+m_1}}{w'^{x_2+m_2+1} w^{x_1+m_1+1}}\frac{{\rm d}z'{\rm d}w'{\rm d}z{\rm d}w}{(w-z)(w'-z')}.
\end{multline}
The location of the contours depends on wether $m_1<m_2$, $m_1=m_2$ or $m_1>m_2$. Let us first assume that $m_1<m_2$.  Then the contour $\Gamma_0$ goes around the pole $w=0$,  the contour $\Gamma_{1,2}$ goes around $z=1$ and $z=2$ and $\Gamma_0$,  the contour $\Gamma_{1,2}'$ goes around  $z'=1$ and $z'=2$ and $\Gamma_0'$  goes around  $w'=0$ and $\Gamma_{1,2}'$. 

There are two ways to compute $R$. By \eqref{eq:defR}, \eqref{eq:symR} and the fact that $m_1\neq m_2$, we can either compute $\sum_{x_1\geq y_1}\sum_{x_2<y_2}$ or $\sum_{x_1< y_1} \sum_{x_2\geq y_2}$ over the terms in \eqref{eq:productofkernels}. We choose to use the second way of summation. In that case,  we make sure that   $|w'|>|z|$ for $w\in \Gamma_0'$ and $z\in \Gamma_{1,2}$ and  $|w|<|z'|$ for $w\in \Gamma_0$ and $z'\in \Gamma_{1,2}'$.  Inserting \eqref{eq:productofkernels} into \eqref{eq:defR},  changing the order of summation and integration, and using
 \begin{align}\label{eq:lemq2}
 \frac{1}{w}\sum_{x_1< y_1} \frac{z'^{x_1}}{w^{x_1}} =\frac{1}{z'-w}, \ \text{and} \   \frac{1}{w}\sum_{x_2\geq y_2} \frac{z^{x_2}}{w'^{x_2}} =\frac{1}{w'-z},
 \end{align} 
gives  the statement in case $m_1<m_2$, with the location of  the contours as in Figure~\ref{fig:quadR}.  

The situation $m_1>m_2$ follows by the same arguments or by the symmetry in \eqref{eq:symR}.  In case $m_1=m_2$, we do not have \eqref{eq:symR}. As a result we can not choose between the two different ways of summing \eqref{eq:productofkernels}. The precise way of summing is given \eqref{eq:defR} and depends on wether $y_1<y_2$ or $y_1\geq y_2$. The statement then follows by the same  reasoning as above for the case $m_1<m_2$.
\end{proof}
\begin{remark}
The locations of the contours in Proposition are chosen for convenience in the proof of Proposition \ref{prop:asymptR1}. There are other deformations of the contours possible  such that \eqref{eq:quadR} holds. 
\end{remark}

\begin{figure}[t]
\centering
\includegraphics[width=0.4\textwidth,height=0.18\textheight]{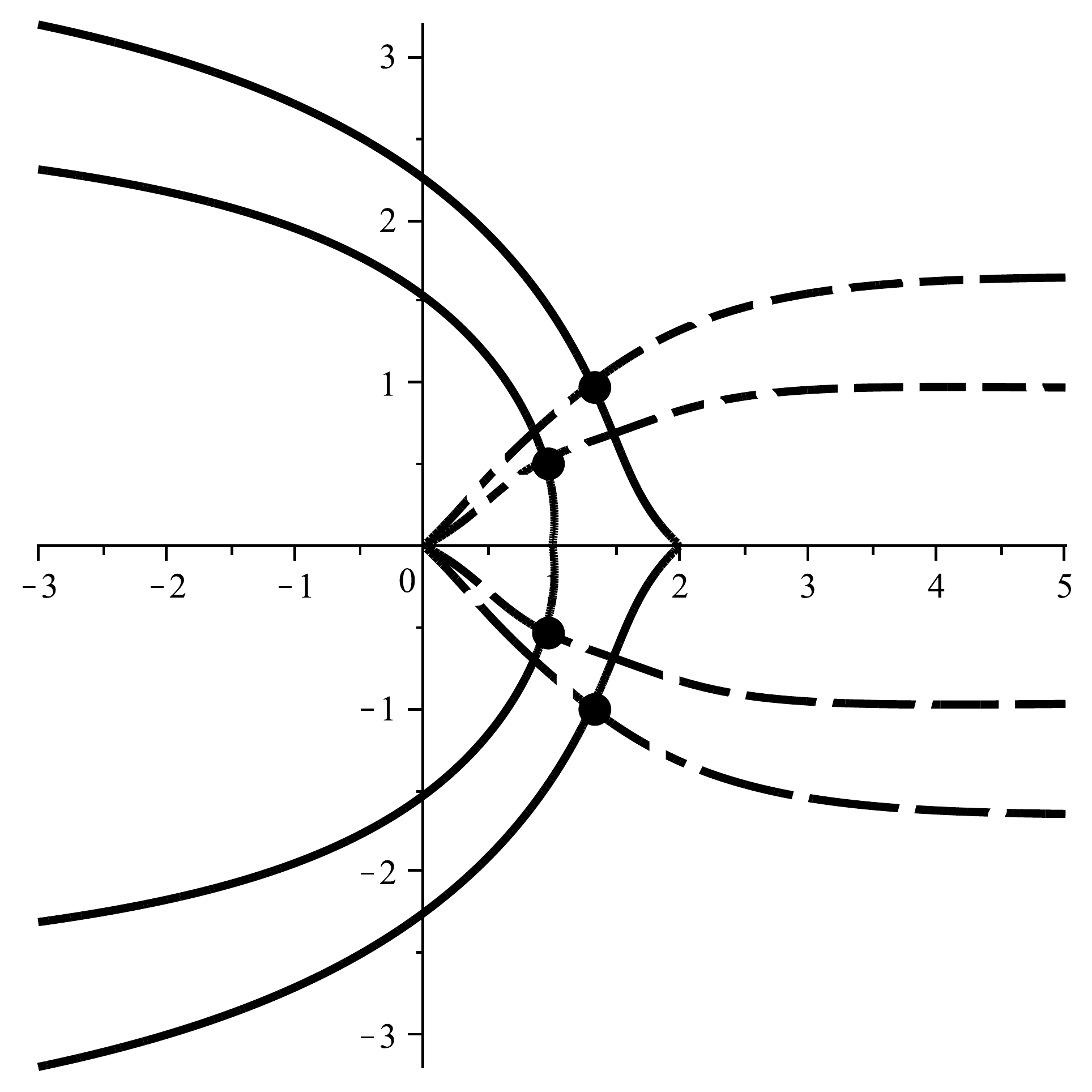}
\includegraphics[width=0.4\textwidth,height=0.18\textheight]{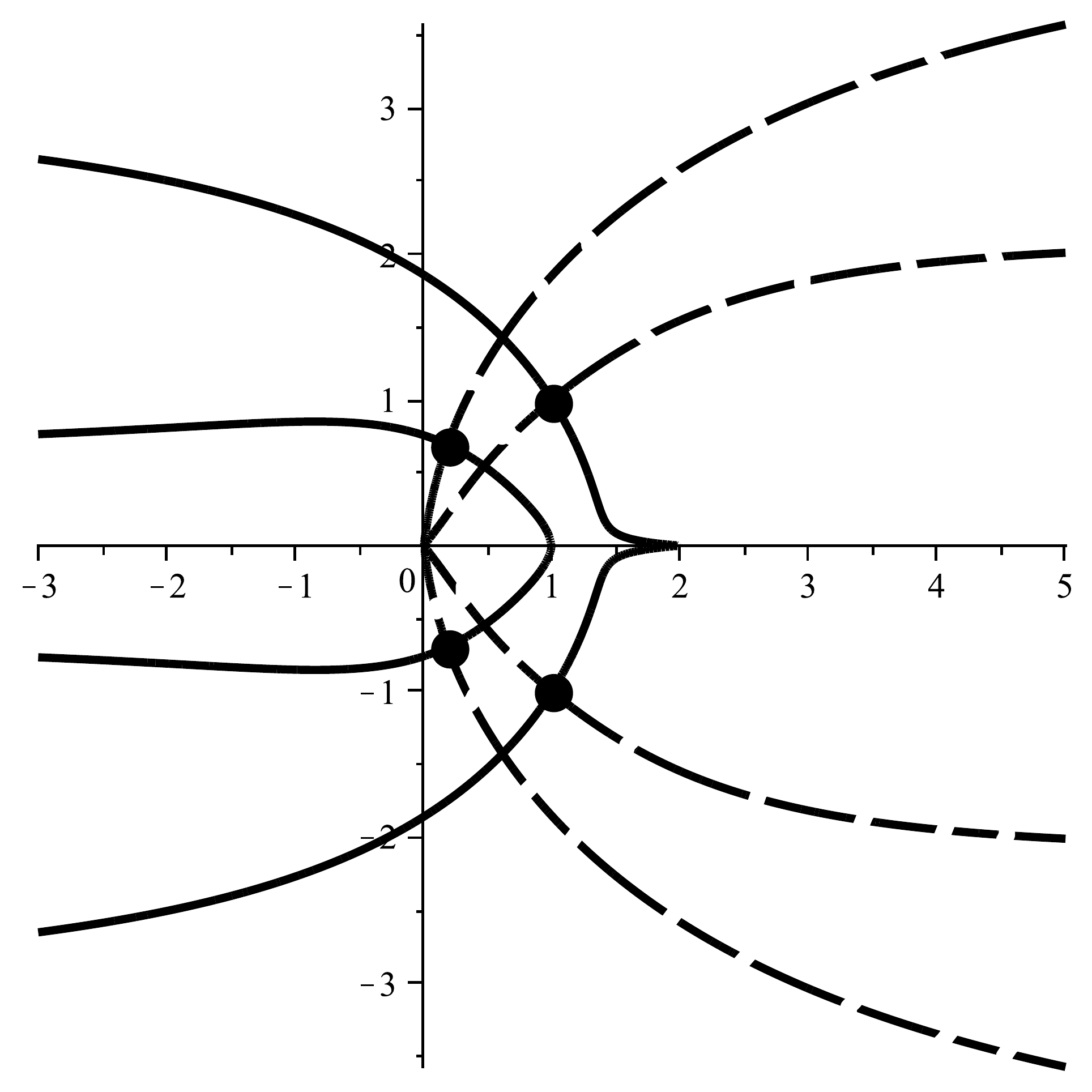}
\caption{Possible configuration of steepest descent/ascent paths for $\Re F_j$ leaving from the saddle points $\Omega_j$ (dots).}
\label{fig:quadRsteep}
 \end{figure}

We are now ready to prove Proposition \ref{prop:asymptR1}.
\begin{proof}[Proof of Proposition \ref{prop:asymptR1}]
The proof of \eqref{eq:limR1} is an elaborate exercise in steepest descent arguments on the quadruple integral representation in \eqref{eq:quadR}, using the same type of arguments as in Section \ref{sec:asymptoticsK}.  

Let us first present the main strategy.  As in Section \ref{sec:asymptoticsK},  we deform the contours  of integration  to paths of steep(est) descent/ascent for $\Re F_j$ leaving from $\Omega_j$ and $\overline \Omega_j$ (for a discussion on these paths see Section \ref{sec:asymptoticsK}.1). Each time we deform one of the contours, we possibly pick up a residue due to the intersection with another contour. By repeating this a number of times we are left with several integrals, each of which is an integral over paths of steep(est) descent/ascent of the integrands. Whenever the integrand contains an exponential, the integral tends to $0$ as $L\to \infty$. There is only one  integral that does not contain an exponential and  this gives the main contribution and the right hand side of \eqref{eq:limR1}. 

Now let us discuss this procedure in more detail for  the case $m_1<m_2$. The case $m_1>m_2$ follows by precisely the same arguments but with $\Gamma_*$ replaced by $\Gamma_*'$ and vice versa. In case $m_1=m_2$, the initial contours for $R$ are slightly different, but this does not lead to essentially new complications. 

We set $m_1<m_2$ and start with the contours as in Figure~\ref{fig:quadR}. Let $\Gamma_0$  first be a very small contour around the origin and $\Gamma_0'$ a very large contour close to infinity far away from the other contours. Deform the contours $\Gamma_{1,2}'$ and $\Gamma_{1,2}$ to paths of steep (not necessarily steepest) ascent leaving from $\Omega_1$ and $\Omega_2$ and their conjugates, such that they do not intersect with $\Gamma_0$ and $\Gamma_0'$ (yet). Then deform the contours $\Gamma_0$ and  $\Gamma_0'$ to paths of steep descent leaving from $\Omega_1$ and $\Omega_2$ and their conjugates. By doing so, we pick up several residues, and are left with quadruple, triple and double integrals. In some of these integrals, we need additional deformations, but first we will collect the various terms that we obtained.

For illustration purposes, let us also assume that for the deformed contours we still have that $\Gamma_0'$ goes around $\Gamma_0$ and that $\Gamma_{1,2}$ goes around $\Gamma_{1,2}'$ (the other situations follow by  similar arguments). See also the left picture in Figure \ref{fig:quadRsteep}, where $\Omega_2$ is the top saddle point. In this case we  are left   with  seven integrals: one quadruple integral (the same as in \eqref{eq:quadR} but now over paths of steep ascent/descent), four triple integrals over contours that are indicated in Figure~\ref{fig:finaltriple} and two double integrals given in Figure~\ref{fig:finaldouble}. 
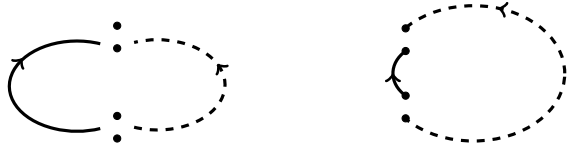
\begin{figure}[t]
\centering{
\begin{tikzpicture}[xscale=0.45,yscale=0.3]
\draw[postaction={decorate},decoration={
  markings,
  mark=at position 0.6 with {\arrow{>}}},very thick,dashed] (0,0) arc (-110:110:2);
\draw[postaction={decorate},decoration={
  markings,
  mark=at position 0.4 with {\arrow{<}}},very thick,rotate around = {180:(-1,3.7)}]  (-1,3.7) arc (-110:110:2);
\filldraw (-0.5,4.5) ellipse (3pt and 4.5pt);
\filldraw (-0.5,3.5)  ellipse (3pt and 4.5pt);
\filldraw (-0.5,0.5)  ellipse (3pt and 4.5pt);
\filldraw (-0.5,-0.5) ellipse (3pt and 4.5pt);
\end{tikzpicture} 
\begin{tikzpicture}[scale=0.45]
\draw[postaction={decorate},decoration={
  markings,
  mark=at position 0.2 with {\arrow{>}}},very thick,dashed]  (2.5,0) ellipse (2.7 and 2);
\filldraw[color=white,very thick]  (-1.5,0) ellipse (2.7 and 2);
 \draw[postaction={decorate},decoration={
  markings,
  mark=at position 0.5 with {\arrow{>}}},very thick] (0.5,-0.66) .. controls (0,-.3) and (0,.3) .. (0.5,0.66);
\filldraw (0.5,1.33) circle (3pt);
\filldraw (0.5,0.66) circle (3pt);
\filldraw (0.5,-1.33) circle (3pt);
\filldraw (0.5,-0.66) circle (3pt);
\end{tikzpicture}
\caption{The paths of integration for the two double integrals after deforming the contours.}
\label{fig:finaldouble}
}
\end{figure}

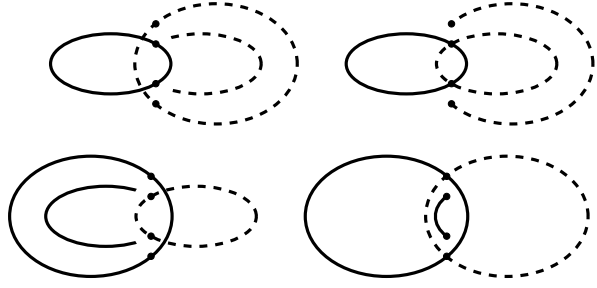
\begin{figure}[t]
\centering{


\begin{tikzpicture}[scale=0.4]
\draw[very thick,dashed]  (2,0) ellipse (2 and 1);
\filldraw[color=white,very thick]  (-1.5,0) ellipse (2.7 and 2);
\draw[very thick,dashed]  (2.5,0) ellipse (2.7 and 2);
  \draw[very thick] (-1,0) ellipse (2 and 1);

\filldraw (0.5,1.33) circle (3pt);
\filldraw (0.5,0.66) circle (3pt);
\filldraw (0.5,-1.33) circle (3pt);
\filldraw (0.5,-0.66) circle (3pt);

\end{tikzpicture}
\begin{tikzpicture}[scale=0.4]
\draw[very thick,dashed]  (2.5,0) ellipse (2.7 and 2);

\filldraw[color=white,very thick]  (-1.5,0) ellipse (2.7 and 2);
  \draw[very thick] (-1,0) ellipse (2 and 1);
  \draw[very thick,dashed]  (2,0) ellipse (2 and 1);

\filldraw (0.5,1.33) circle (3pt);
\filldraw (0.5,0.66) circle (3pt);
\filldraw (0.5,-1.33) circle (3pt);
\filldraw (0.5,-0.66) circle (3pt);
\end{tikzpicture}
\vskip 10pt
\begin{tikzpicture}[scale=0.4]
 \draw[very thick] (-1,0) ellipse (2 and 1);
\filldraw[color=white,very thick]   (2.5,0) ellipse (2.7 and 2);
\draw[very thick] (-1.5,0) ellipse (2.7 and 2);
   \draw[very thick,dashed]  (2,0) ellipse (2 and 1);
\filldraw (0.5,1.33) circle (3pt);
\filldraw (0.5,0.66) circle (3pt);
\filldraw (0.5,-1.33) circle (3pt);
\filldraw (0.5,-0.66) circle (3pt);
\end{tikzpicture}
\begin{tikzpicture}[scale=0.4]
\draw[very thick,dashed]   (2.5,0) ellipse (2.7 and 2);
\draw[very thick] (-1.5,0) ellipse (2.7 and 2);
  \draw[very thick] (0.5,-0.66) .. controls (0,-.3) and (0,.3) .. (0.5,0.66);
\filldraw (0.5,1.33) circle (3pt);
\filldraw (0.5,0.66) circle (3pt);
\filldraw (0.5,-1.33) circle (3pt);
\filldraw (0.5,-0.66) circle (3pt);
\end{tikzpicture}
}
\caption{The contours for the four triple integrals after deforming the contours.}
\label{fig:finaltriple}
\end{figure}
 
 We discuss the seven integrals that are obtained in this way, starting  with the double integrals in Figure~\ref{fig:finaldouble}. The right picture represents the integral 
\begin{align}\label{eq:finaldouble}
\frac{1}{(2\pi {\rm i})^2}\int_{\overline{\Omega}_1}^{\Omega_1} \int_{\overline{\Omega}_2}^{\Omega_2} \frac{1}{(z-z')^2} {\rm d}z {\rm d}z'.
\end{align}
The integral is easily computed and gives the Green's function at the right-hand side of \eqref{eq:limR1}.
 The double integral over the contours in left picture of Figure~\ref{fig:finaldouble} is given by
 \begin{align}\label{eq:finaldouble2}
\frac{1}{(2\pi {\rm i})^2}\int_{\overline{\eta}_1}^{\eta_1} \int_{\overline{\eta}_2}^{\eta_2} \frac{{\rm e}^{L (F_1(z)+F_2(z'))}}{{\rm e}^{L(F_2(z)+F_1(z'))}} \frac{1}{(z-z')^2} {\rm d}z' {\rm d}z,
\end{align}
where $\eta_1$ and $\eta_2$ are  the point of intersection in the upper half plane of $\Gamma_0$ and $\Gamma_{1,2}$, and the contours $\Gamma_0'$ and $\Gamma_{1,2}'$ respectively. Deform the paths of integration to arcs of circles centered at the origin with radius $|\eta_j|$. As shown in the proof of Lemma \ref{lem:I1}, these arcs are paths of steep descent and ascent for $\Re (F_1-F_2)$ leaving from $\eta_1$ and $\eta_2$  (here we use the fact that $m_1<m_2$). Moreover, since the $\eta_j$ are intersection points of contours of steep descent/ascent for $\Re F_j$, we also have
\begin{align}
\Re\left( F_1(\eta_1)-F_2(\eta_1)\right)< \Re\left( F_1(\Omega_1)-F_2(\Omega_2)\right)<\Re\left( F_1(\eta_2)-F_2(\eta_2)\right).\end{align} As a result, we see that the integral in \eqref{eq:finaldouble2}
 is exponentially small as $L\to \infty$, when $\Omega_1$ and $\Omega_2$ are sufficiently far.

 In the next step, we analyze the quadruple integral 
 \begin{align}\label{eq:finalqdruple}
 \frac{-1}{(2\pi{\rm i})^4} \oint_{\Gamma_0} \oint_{\Gamma_{1,2}}  \oint_{\Gamma_0'}  \oint_{\Gamma_{1,2}' }\frac{{\rm e}^{L(F_1(w)+F_2(w'))}}{{\rm e}^{L(F_2(z)+F_1(z'))}} 
\frac{{\rm d}z' {\rm d}w' {\rm d}z {\rm d}w}{(w-z)(w-z')(w'-z)(w'-z')},
\end{align}
which is the same integral as in  \eqref{eq:quadR} but now with deformed contours.  Since the deformed contours are of steep descent/ascent and the functions $F_j$ appear both the numerator and denominator, the integrand is bounded. In fact, the main contribution, comes from small neighborhoods around the saddle points.  As in  Section 6.3, we introduce the local variables as in \eqref{eq:localvar} around the saddle points. For example, near $\Omega_1$ and $\Omega_2$ we introduce
\begin{align}
\begin{split}
w=\Omega_1+ \frac{s}{L^{1/2} (-F_1''(\Omega_1))^{1/2}},\qquad  z=\Omega_2+ \frac{t}{L^{1/2}(F_2''(\Omega_2))^{1/2}},\\
 w'=\Omega_2+ \frac{s'}{L^{1/2}(-F_2''(\Omega_2))^{1/2}},\qquad  z'=\Omega_1+ \frac{t'}{L^{1/2}(F_1''(\Omega_1))^{1/2}}. 
\end{split}\end{align}    The scaling by $L^{1/2}$ in the local variable, implies that the integrals  tend to zero as  $L\to \infty$ (as was the case for the double integral for $K_G$ in \eqref{eq:IG21}). To this end, note that
\begin{multline}\label{last} \frac{{\rm d}w {\rm d}z {\rm d}w' {\rm d}z'}{(w-z)(w-z')(w'-z)(w'-z')}
=\frac{\pm 1 }{L F_1''(\Omega) F''(\Omega_2)} \frac{1}{(s\pm {\rm i} t')(s'\pm{\rm i}t)}\\
\times \frac{1}{ \Omega_1-\Omega_2+\frac{s}{L^{1/2} (-F_1''(\Omega_1))^{1/2}}-\frac{t}{L^{1/2} (F_2''(\Omega_2))^{1/2}}}\\
\times\frac{1}{\Omega_2-\Omega_1+\frac{s'}{L^{1/2} (F_2''(\Omega_2))^{1/2}}-\frac{t'}{L^{1/2} (F_1''(\Omega_1))^{1/2}}}.\end{multline}
The sign depends on the choice of the square roots, but is irrelevant in this discussion. Since by assumption the points $\Omega_1$ and $\Omega_2$ are sufficiently far apart and the fact that the right-hand side is integrable (despite the singularities), we indeed see that the quadruple integral is of order $\OO(L^{-1})$ as $L\to \infty$.  Moreover, the convergence is uniform on compact subsets of $\mathcal D$.

We now come to the triple integrals presented in Figure~\ref{fig:finaltriple}.  Let us start with the top right case. In that situation the triple integral is given by 
\begin{multline}\label{localvarR}
\oint_{\Gamma_0} \oint_{\Gamma_{1,2}'} \int_{\overline{\Omega}_2}^{\Omega_2}
\frac{{\rm e}^{L F_1(w)}}{{\rm e}^{{LF_1(z')}}} \frac{ {\rm d}z {\rm d} z'{\rm d}w}{(w-z)(w-z')(z-z')} 
\end{multline}
 The integral over the path from $\overline{\Omega}_2$ to $\Omega_2$ can be explicitly computed.  As for   the integrals over $\Gamma_0$ and $\Gamma_{1,2}'$, the main contribution comes again from small neighborhoods near the saddle points.  By introducing local variables as in \eqref{localvarR} and arguing as before,   it is apparent that the integral vanishes as $L\to \infty$ uniformly on compact subsets of $\mathcal D$. The same can be done for the situation in the lower right picture. 
 
  For the triple integral over contours as in the top left picture of Figure~\ref{fig:finaltriple} we have
 \begin{align}
 \oint_{\Gamma_0} \oint_{\Gamma_{1,2}} \int_{\overline{\eta}_2}^{\eta_2}
\frac{{\rm e}^{L (F_1(w)+F_2(z'))}}{{\rm e}^{{L(F_1(z')+F_2(z))}}} \frac{ {\rm d}z' {\rm d} z{\rm d}w}{(w-z)(w-z')(z-z')} 
 \end{align}
 In this triple integral we deform the path from $\overline{\eta}_2$ to $\eta_2$ to an arc of circle centered at the origin which is path of steep descent for $\Re (F_2-F_1)$ for the case $m_1<m_2$ (as mentioned earlier in the treatment of \eqref{eq:finaldouble2}. Hence, it is exponentially small as $L\to \infty$. Note that in the latter deformation we possibly pick up a residue again, which results in an additional  double integral. However, since the contours are paths of steep descent/ascent the contribution of this term, possibly after an additional deformation, can be shown to be negligible.

Concluding, from the seven multiple integrals obtained after deforming the contours, the leading term in the asymptotic expansion for $R$ comes from the  double integral \eqref{eq:finaldouble}. This integral equals the right-hand side of \eqref{eq:limR1} and we proved the statement. \\

 Finally, we come to \eqref{eq:boundR}. If $(x_j,m_j)$ are sufficiently far apart, then this bound follows from \eqref{eq:limR1}. If on the other hand the points are close, then we need to deal with singularities in the integrals above. For example, the double integral that is over the contours in the left picture of Figure~\ref{fig:finaldouble}, is no longer exponentially small. In fact, both double integrals \eqref{eq:finaldouble} and \eqref{eq:finaldouble2} grow logarithmically with $L$ as $L\to \infty$,  in case the points $\Omega_j$ come closer and closer. The fact of the matter is that  in all integrals obtained after deformation, the logarithmic behavior is however the worst that can happen.  This can be  checked by introducing local variables \eqref{localvarR} in the triple and quadruple integrals.  For example, in the quadruple integral \eqref{eq:finalqdruple} the integrand has a singular term \eqref{last}. As long as $\Omega_1\neq \Omega_2$, then this is still integrable and the result is of order $\OO(\log(\Omega_1-\Omega_2))$ as $\Omega_1-\Omega_2\to 0$.  Hence for $m_1<m_2$ we obtain \eqref{eq:boundR}.
 
The other situations $m_1=m_2$ and $m_1>m_2$ can be dealt with in a similar way. However, care should be taken in  case $(x_1,m_1)=(x_2,m_2)$ (and hence $\Omega_1=\Omega_2$). In that case, deforming the contours leads to divergent  integrals. The way around this is to perturb one of the saddle points  (and the paths of steep descent/ascent) with a term of $\OO(L^{-1})$ and argue as above.  Then we again obtain the logarithmic growth as $L\to \infty$ and this shows that \eqref{eq:boundR} holds for all points $(x_j,m_j)$ in the bulk of $ L \mathcal D$. \end{proof}

\end{document}